\documentclass{article}
\usepackage{graphicx} 

\usepackage[text={17.1cm,24.6cm},centering]{geometry} 
\usepackage[utf8]{inputenc}
\usepackage{authblk}
\usepackage{blkarray}
\usepackage{cite}
\usepackage{comment}
\usepackage{amsfonts,amsmath,amssymb,amsthm}
\usepackage{mathrsfs,mathtools,bm}
\usepackage{graphicx,subcaption,caption,float,xcolor}
\usepackage{enumitem}
\usepackage{chngcntr}
\usepackage[breaklinks]{hyperref} 
\usepackage[normalem]{ulem}

\usepackage{tikz}
\usetikzlibrary{calc}

\usepackage{caption}

\usepackage{xcolor}

\newtheorem{theo}{Theorem}[section]

\newtheorem{rem}{Remark}[section]
\newtheorem*{theo*}{Theorem}
\newtheorem{lemm}[theo]{Lemma}
\newtheorem{prop}[theo]{Proposition}

\newtheorem{notation}[theo]{Notation}

\numberwithin{equation}{section}
\newtheorem{remax}[theo]{Remark}
\newenvironment{rema}
  {\pushQED{\qed}\remax}
  {\popQED\endremax}
  
\theoremstyle{definition}

\newcommand{\cH}{\mathcal{H}}
\newcommand{\cC}{\mathcal{C}}
\newcommand{\R}{\mathbb{R}}

\newcommand{\N}{\mathbb{N}}

\newcommand{\bra}[1]{\langle\,{#1}}
\newcommand{\ket}[1]{\mid{#1}\,\rangle}

\def\rr{{\boldsymbol{\rho}}}
\def\orr{\overline{{\boldsymbol{\rho}}}}
\def\ox{{\overline{x}}}

\def\oalpha{\overline{{\alpha}}}
\def\obeta{\overline{{\beta}}}
\def\oN{\overline{{N}}}

\newcommand{\argu}[3]{\left(\begin{array}{c} #1\\#2\end{array} ; #3\right)}

\title{The dynamical algebra of the generic superintegrable model\\ on the two-sphere}
\author{
Nicolas Crampé\textsuperscript{$1$}\footnote{E-mail: nicolas.crampe@cnrs.fr}~, 
Quentin Labriet\textsuperscript{$2$}\footnote{E-mail: quentin.labriet@umontreal.ca}~, 
Lucia Morey\textsuperscript{$2$}\footnote{E-mail: lucia.morey@umontreal.ca}~, 
Satoshi Tsujimoto\textsuperscript{$3$}\footnote{E-mail: tsujimoto.satoshi.5s@kyoto-u.jp}~,\\ 
Luc Vinet\textsuperscript{$2,4$}\footnote{E-mail: luc.vinet@umontreal.ca}~
\vspace{0.2cm}, 
Alexei Zhedanov\textsuperscript{$5$}\footnote{E-mail: zhedanov@yahoo.com}\vspace{0.2cm}\\
\textsuperscript{$1$}
\small Institut Denis-Poisson CNRS/UMR 7013 - Universit\'e de Tours \\
\small Parc de Grammont, 37200 Tours, France.\vspace{0.2cm}\\
\textsuperscript{$2$}
\small Centre de Recherches Math\'ematiques, Universit\'e de Montr\'eal, P.O. Box 6128, \\
\small Centre-ville Station, Montr\'eal (Qu\'ebec), H3C 3J7, Canada.\vspace{0.2cm}\\
\textsuperscript{$3$}
\small Graduate School of Informatics, Kyoto University,
Yoshida-Honmachi, Kyoto, Japan 606-8501.\vspace{0.2cm}\\
\textsuperscript{$4$}
\small IVADO, Montr\'eal (Qu\'ebec), H2S 3H1, Canada.\vspace{0.2cm} \\
\textsuperscript{$5$}
\small
Leonhard Euler International Mathematical Institute, Saint Petersburg, Russian Federation.
}

\date{May 2025}

\begin{document}

\maketitle

\begin{center}
\begin{minipage}{12cm}
\begin{center}
{\bf Abstract}\\
\end{center}  
The rank two Jacobi algebra $\mathfrak{J}_2$ is identified as the dynamical algebra of the generic quadratic superintegrable model on the two-sphere. The physical representation of this algebra is obtained from its embedding in $\mathfrak{su}(1,1)^{\otimes 3}$. The exact solution of the model is derived algebraically from this representation. The wavefunctions are found to be expressed in terms of two-variable Jacobi polynomials whose characterization is a by-product of the algebraic treatment of the model. 
\end{minipage}
\end{center}

\medskip
\section{Introduction}
Superintegrable systems \cite{miller2013classical}, like the non-relativistic Hydrogen atom \cite{pauli1926wasserstoffspektrum},  are the hallmark of symmetries. Apart from their intrinsic interest as models, they offer remarkable laboratories for the study of algebraic structures poised to describe invariance properties. A quantum model in $d$ dimensions is maximally superintegrable if it possesses $2d-1$ algebraically independent constants of motion $H_k, \; k=1,\dots 2d-1$, including its Hamiltonian $H$ such that 
\begin{equation}
    [H,H_k]=0, \text{ for all}\ k.
\end{equation}
These operators generate the \textit{symmetry algebra} of the system which is necessarily non-abelian. The dimensions of its unitary irreducible representations account for the degeneracies of the energy levels. One calls \textit{dynamical} an algebra that contains the symmetry algebra as a subalgebra and that has the span of the whole set of the energy eigenstates of the Hamiltonian as an irreducible (typically infinite dimensional) module \cite{bohm1988dynamical}. Clearly, this dynamical or spectrum-generating algebra has generators that do not commute with the Hamiltonian. The simplest example of such a dynamical algebra is the Heisenberg one that allows to connect any state vector of the harmonic oscillator to another through the repeated application of the annihilation or creation operators. For the bound state sector of the Hydrogen atom, the symmetry algebra is $\mathfrak{so}(4)$ \cite{baym2018lectures}, and the dynamical one is the conformal algebra $\mathfrak{so}(4,2)$ \cite{barut1971so}, \cite{d1985spectrum}.

The so-called generic superintegrable model on the sphere $S^n$ has been the object of much attention (see for instance \cite{kalnins2007wilson},\cite{kalnins2011two}, \cite{kalnins2013contractions}, \cite{GVZ2014},\cite{genest2014racah}, \cite{genest2014generic}, \cite{post2015racah}, \cite{miller2015quasi}, \cite{Iliev17}, \cite{de2017higher}, \cite{Iliev18}, \cite{post2024racah}). While the Racah algebras of rank $n-1$ have been identified as symmetry algebras, and hidden Lie algebra structures underlying exact solvability have been exhibited in terms of $\mathfrak{gl}_n$-type algebras \cite{TurbinerMiller2014}, \cite{miller2015quasi}, the dynamical algebra of these systems—namely an algebra generating the full spectrum and organizing both differential and recurrence structures—has, to the best of our knowledge, not been identified. We shall focus on the 2-dimensional system which leads to all the other second order non-relativistic superintegrable models on spaces with positive signature through specializations and contractions \cite{kalnins2013contractions}, and add this important missing element to its description. This will feature the rank two Jacobi algebra \cite{crampe2025rank}, \cite{crampe2026algebraic}, in the dynamical algebra role with the superintegrable model providing a natural framework for the construction of its physically relevant representation. It will be seen that the identification of this dynamical algebra allows for an algebraic solution of the associated Schr\"odinger equation from the knowledge of its representation. Moreover, this study will be an occasion to discuss and further explore algebras of Askey--Wilson types. 
The explicit definition of the model follows.

\subsection{The superintegrable system on the two-sphere}\label{sec:GenericSystem}

The generic second-order quantum superintegrable systems on the two-sphere \cite{kalnins2007wilson} is governed by the following Hamiltonian 
\begin{equation}
\label{eq:Hamiltonian}
\mathcal{H}=\Hat{\mathcal{J}}_1^2+\Hat{\mathcal{J}}_2^2+\Hat{\mathcal{J}}_3^2+\frac{a^2-\frac{1}{4}}{x_1^2}+\frac{b^2-\frac{1}{4}}{x_2^2}+\frac{c^2-\frac{1}{4}}{x_3^2},
\end{equation}
where $a, b, c$ are real parameters and the unit two-sphere is embedded in $\mathbb{R}^3$ with coordinates $x_1, x_2, x_3$ by imposing $x_1^2+x_2^2+x_3^2=1$. Using the manifest symmetry under the permutations of these three coordinates, we restrict to the octant with $x_1,x_2,x_3>0$.
The symbols $\Hat{\mathcal{J}}_k$, with $k=1,2,3$ denote the standard angular momentum operators:
\begin{equation}
\label{eq:angular-momentum-generators}
\Hat{\mathcal{J}}_k=-i \epsilon _{klm} x_l\partial_m,
\end{equation}
where $\partial_k$ denotes the partial derivative with respect to $x_k$, the sum over repeated indices taking the values 1,2 and 3 is understood and $\epsilon_{klm}$ is the Levi-Civita antisymmetric tensor (with $\epsilon _{123} = 1$). 
As already mentioned all second-order superintegrable models in two dimensions can be obtained as limits of this generic 3-parameter system  \cite{kalnins2013contractions}. 

\subsection{Outline}
We shall explain in the following that this model has a degenerate spectrum and is indeed maximally superintegrable. This will be done in the next section by reviewing results \cite{GVZ2014} establishing that the Hamiltonian $\mathcal{H}$ can be obtained as the total Casimir element in a differential realization of the three-fold tensor product of $\mathfrak{su}(1,1)$. We shall see in Section \ref{sec:3} that this framework is conducive to the identification of the constants of motion and of operators that connect all the energy levels thus providing the generators of the symmetry and dynamical algebras. The $\mathfrak{su}(1,1)^{\otimes 3}$ framework will allow to obtain straightforwardly the structure relations of the dynamical algebra and to recover those of its symmetry subalgebra. These algebras will be respectively recognized to be the rank two Jacobi algebra $\mathfrak{J}_2$\cite{crampe2025rank}, \cite{crampe2026algebraic} and the rank one Racah algebra $\mathfrak{R}_1$ \cite{kalnins2007wilson},\cite{GVZ2014}, \cite{Iliev18}. The embedding of $\mathfrak{J}_2$ in $\mathfrak{su}(1,1)^{\otimes 3}$ will be used in Sections \ref{sec:4} and \ref{sec:5} to construct the physically relevant representation over the $\mathfrak{su}(1,1)$-module made out of three positive discrete series. The representation basis will be given by the eigenstates of two generators that correspond to the Hamiltonian and one constant of motion in the differential realization. The restriction to the unit two-sphere will be performed by imposing a constraint $\textit{à la}$ Dirac. All this will require calling upon the contiguity relations of families of orthogonal polynomials and in particular those entering in the definition of the Clebsch--Gordan and Racah coefficients of $\mathfrak{su}(1,1)$. Equipped with this knowledge, it will then be explained in Section \ref{sec:eigenstateH}, how the exact solution of the model can be obtained in a purely algebraic fashion by using the representation of $\mathfrak{J}_2$. This will feature the two-variable Jacobi polynomials in terms of which the wavefunctions will be expressed. The differential equations, recurrence and structure relations of these orthogonal polynomials will be found as a by-product. A summary and an outlook will form the Conclusion. Two appendices will complete the paper: Appendix \ref{app:A} will collect the contiguity relations of certain families of orthogonal polynomials and the complete list of the defining relations of the rank two Jacobi algebra is the content of Appendix \ref{sec:Rank2Jacobi}.


\section{The generic superintegrable model and $\mathfrak{su}(1,1)$}

As indicated in the introduction the non-compact Lie algebra $\mathfrak{su}(1,1)$ will play a central role in the present study. We shall now recall how the Hamiltonian $\mathcal{H}$ can be identified as the total Casimir element in the tensor product of three positive discrete series representations of $\mathfrak{su}(1,1)$ as first shown in \cite{GVZ2014}. This will permit to readily determine the spectrum of $\mathcal{H}$ and to see that it is degenerate.

\subsection{Positive discrete series representations of $\mathfrak{su}(1,1)$ }\label{sec:Positive discrete series}

The Lie algebra $\mathfrak{su}(1,1)$ is generated by three elements $J_0,$ $J_{\pm}$ satisfying the commutation relations 
\begin{equation}
[J_0,J_{\pm}]=\pm J_{\pm}, \quad [J_+,J_-]=-2J_0    .\label{su11cr}
\end{equation}
The Casimir operator, which generates the center of its universal enveloping algebra, is given by
\begin{equation}
    C=J_0^2-J_0-J_+J_-. 
\end{equation}

The positive discrete series representations of $\mathfrak{su}(1,1)$  with real parameter $\nu>0$ acting on $\ket{\nu;n}$ for $n\in \N$ is defined by
\begin{align}
J_0 \ket{\nu;n}&=(\nu + n)\ket{\nu;n}, \\
J_+ \ket{\nu;n}&=\ \ket{\nu ;n+1},\\
J_-\ket{\nu;n}&=n(n+2\nu-1)\ket{\nu;n-1}.
\end{align}
The basis $\ket{\nu;n}$ is taken to be orthogonal
with 

\begin{equation}
    \langle \nu,m |\nu, n\rangle = n!(2\nu)_n \delta_{mn}, \label{orthbasis}
\end{equation}
making the above representation unitary. The Casimir operator $C$ has the eigenvalue $\nu(\nu-1)$ in these positive discrete series representations. 

\subsection{Three-fold tensor product}
Consider now  the algebra $\mathfrak{su}(1,1)^{\otimes 3}$ and define for any $T\in \mathfrak{su}(1,1)$ the embeddings
\begin{equation}
    T^{(1)}=T\otimes 1\otimes 1,\quad T^{(2)}=1\otimes T\otimes 1,\quad T^{(3)}=1\otimes 1\otimes T.
\end{equation}
These three embeddings produce a fourth one, called the diagonal embedding, which is defined as follows
\begin{equation}
    T^{(123)}=T^{(1)}+T^{(2)}+T^{(3)}.
\end{equation}
The total Casimir element $C^{(123)}={J_0^{(123)}}^2-J_0^{(123)}-J_+^{(123)}J_-^{(123)}$ is seen to have the following expression
\begin{equation}
    C^{(123)}=C^{(12)}+C^{(23)}+C^{(13)}-C^{(1)}-C^{(2)}-C^{(3)}, \label{c123cij}
\end{equation}
where $C^{(i)}={J_0^{(i)}}^2-J_0^{(i)}-J_+^{(i)}J_-^{(i)}$ are the individual Casimir operators and $C^{(ij)}$ are the intermediate Casimir operators given by
\begin{equation}
    C^{(ij)}=2J_0^{(i)}J_0^{(j)}-(J_+^{(i)}J_-^{(j)}+J_+^{(j)}J_-^{(i)})+C^{(i)}+C^{(j)}.
\end{equation}
The total Casimir operator $C^{(123)}$ commutes with all the intermediate Casimir operators $C^{(ij)}$ as well as with all the individual Casimir operators $C^{(i)}$. However the intermediate Casimir operators do not commutes with each other.

\subsection{The Hamiltonian $\mathcal{H}$ as a $\mathfrak{su}(1,1)$ Casimir element}
The connection between the triple tensor product of positive discrete series representations of $\mathfrak{su}(1,1)$ and the generic 3-parameter superintegrable system was established in \cite{GVZ2014} using a diﬀerential realization of $\mathfrak{su}(1,1)$. This relies on the fact that the one-dimensional singular oscillator has $\mathfrak{su}(1,1)$ as dynamical algebra. In the variable $x_i$ (with $i$ meaningless for the moment), the Hamiltonian of this oscillator system reads

\begin{equation}
    \mathcal{J}_0^{(i)}=\frac{1}{4}\left(-\partial_{i}^2+x_i^2+\frac{4\nu_i(\nu_i-1)+\frac{3}{4}}{x_i^2}\right).\label{J0}
\end{equation}
Adjoin to $\mathcal{J}_0^{(i)}$ the two operators $\mathcal{J}_{\pm}^{(i)}$ given by
\begin{equation}
   \mathcal{J}_{\pm}^{(i)}=\frac{1}{4}\left(\partial_{i}^2\mp2x_i\partial_{i}+(x_i^2\mp1)-\frac{4\nu_i(\nu_i-1)+\frac{3}{4}}{x_i^2}\right). \label{j+-}
\end{equation}
It is readily checked that $\mathcal{J}_0^{(i)}$ and $\mathcal{J}_{\pm}^{(i)}$ provide a realization of the $\mathfrak{su}(1,1)$ commutation relations \eqref{su11cr}. Owing to the fact that $\mathcal{J}_0^{(i)}$ can be factorized as follows
\begin{equation}
    \mathcal{J}_0^{(i)} = [\mathcal{A}_{\nu}^{(i)}]^\dagger \mathcal{A}_{\nu}^{(i)} + \nu _i,
\end{equation}
with
\begin{equation}
    \mathcal{A}_{\nu}^{(i)}=\frac{1}{2}\big(-\partial_i + x_i - \frac{2\nu _i - 1/2}{x_i}\big),
\end{equation}
this Hamiltonian is bounded below (and in fact positive for $\nu _i \in \mathbb{R}^+$).
The Casimir operator $\mathcal{C}^{(i)}$ is found to be $\nu_i(\nu_i-1) \mathbb{I}$.
This shows that the realization \eqref{J0}, \eqref{j+-} corresponds to a positive discrete series representation with parameter $\nu_i$. 
The sprectrum of $\mathcal{J}_0^{(i)}$ is hence given by $\nu _i + n, n \in \mathbb{N}$.

Let us now take $i=1,2,3$. Consider then the sum of the corresponding three singular oscillators and the $\mathfrak{su}(1,1)$ representation associated to the diagonal embedding in $\mathfrak{su}(1,1)^{\otimes 3}$ with the variable $x_i$ attached to the $i^{th}$ factor. Choose the $\nu_i$ as follows
\begin{equation}\label{eq:Identification Nu abc}
    \nu_1=\frac{a+1}{2},\ \nu_2=\frac{b+1}{2}\  \text{and}\ \nu_3=\frac{c+1}{2}.
\end{equation}

\begin{notation} 
We shall use the following notations in the rest of the paper:
\begin{enumerate}

    \item 
In keeping with subsection 2.2, we shall write 
\begin{equation}
    \nu_{ij}=\nu_i+\nu_j,\ \  \nu_{123}=\nu_{1}+\nu_2+\nu_3.
\end{equation}

\item 

Furthermore, we shall designate by a standard straight letter the elements of an algebra or their abstract representation and by calligraphic letters their differential realizations as the case may be. 
\end{enumerate}
\end{notation}

Upon using \eqref{J0}, \eqref{j+-}, the full Casimir operator $\cC^{(123)}$ and the intermediate Casimir operators $\cC^{(ij)}$ have the expressions 
\begin{align}
\cC^{(123)}&=\frac{1}{4}\left(\Hat{\mathcal{J}}_1^2+\Hat{\mathcal{J}}_2^2+\Hat{\mathcal{J}}_3^2+(x_1^2+x_2^2+x_3^2)\left(\frac{a^2-\frac{1}{4}}{x_1^2}+\frac{b^2-\frac{1}{4}}{x_2^2}+\frac{c^2-\frac{1}{4}}{x_3^2}\right)-\frac{3}{4}\right),\label{eq:decompositionCasimir}\\
\cC^{(ij)}&=\frac{1}{4}\left(\Hat{\mathcal{J}}_k^2+\frac{(4\nu_i(\nu_i-1)+\frac{3}{4})x_j^2}{x_i^2}+\frac{(4\nu_j(\nu_j-1)+\frac{3}{4})x_i^2}{x_j^2}+4\nu_i(\nu_i-1)+4\nu_j(\nu_j-1)+\frac{1}{2}\right),
\end{align}
where $\Hat{\mathcal{J}}_i$ are the angular momentum operators \eqref{eq:angular-momentum-generators} and the indices are such that $\epsilon_{ijk}\neq 0$.  As expected from \eqref{c123cij}, we check that 
\begin{equation}
    \cC^{(123)}=\cC^{(23)} +\cC^{(13)} +\cC^{(12)} - \cC^{(1)}- \cC^{(2)}- \cC^{(3)}.\label{realc123cij}
\end{equation}
Consider the following element of $\mathfrak{su}(1,1)^{\otimes 3}$
\begin{equation}
X^{(i)}=2J_0^{(i)}+J_{+}^{(i)}+J_{-}^{(i)}, \qquad i=1,2,3.
\end{equation}
In the realization \eqref{J0}, \eqref{j+-}, $X^{(i)}$ is represented by
\begin{equation}
    \mathcal{X}^{(i)}=x_i^2.
\end{equation}
Introduce $X^{(123)}=X^{(1)}+X^{(2)}+X^{(3)}$ which is realized by
\begin{equation}\label{eq:DefinitionX}
\mathcal{X}^{(123)}=\mathcal{X}^{(1)}+\mathcal{X}^{(2)}+\mathcal{X}^{(3)}=x_1^2+x_2^2+x_3^2 = r^2.
\end{equation}
Clearly, the operator $X^{(123)}$ commutes with all the Casimir operators and can thus be simultaneously diagonalized with those elements. Thus on the eigenspace of $\mathcal{X}^{(123)}$ with eigenvalue $r^2=1$, the Hamiltonian \eqref{eq:Hamiltonian} is immediately seen to be given by
\begin{equation}
    \mathcal{H}=4\cC^{(123)}+\frac{3}{4}, \label{HamandCas}
\end{equation}
that is, up to an affine transformation, by the total Casimir of the diagonal embedding in the three-fold product of the realization \eqref{J0} and $\eqref{j+-}$ of $\mathfrak{su}(1,1)$.

It is known that the irreducible components of the triple tensor product of positive discrete series are discrete series representations themselves with parameter $\nu_{123}+n$ with $n>0$. The eigenvalues of the Casimir operator $\cC^{(123)}$ are hence $(\nu_{123}+n)(\nu_{123}+n-1)$ and the spectrum of the Hamiltonian $\cH$ is given by
\begin{equation}\label{eq:Energy}
    E_n=4(\nu_{123}+n)(\nu_{123}+n-1)+\frac{3}{4}=(2(n+1)+a+b+c)^2-\frac{1}{4},
\end{equation}
for $n$ non negative integers. It exhibits degeneracies that can be explained in the $\mathfrak{su}(1,1)$ framework developed so far. This is the object of the next section that will confirm the superintegrable character of $ \mathcal{H}$.


\section{Symmetry and dynamical algebras \label{sec:3}}

We are now in a position to identify and present the algebras that encode the dynamics of the two-dimensional system with Hamiltonian $\mathcal{H}$. 

\subsection{The rank one Racah algebra as the symmetry algebra}

For $\mathcal{H}$ to  be maximally superintegrable it needs to possess three algebraically independent constants of motion including $\mathcal{H}$. Given that the Hamiltonian is essentially  the Casimir element $\cC^{(123)}$  as per \eqref{HamandCas}, the intermediate Casimir operators $\cC^{(ij)}$ that obviously commute with $\cC^{(123)}$ are all constants of motion and will generate a non-abelian algebra since they do not commute among themselves. In view of the relation \eqref{realc123cij}, we see that these constants of motion are not all independent and we shall take $\cC^{(123)}$ (i.e. $\mathcal{H}$), $\cC^{(12)}$ and $\cC^{(23)}$ as the three independent ones confirming that $\mathcal{H}$ is maximally superintegrable.

It is known that the algebra generated by the intermediate Casimir operators $C^{(ij)}$ is the rank one Racah algebra $\mathfrak{R}_1$ \footnote{Technically, it is a quotient by the relations specifying the values of the Casimir elements (see \cite{CFGPaRL21}) but we shall not be concerned with these nomenclature distinctions here.}. The presentation of the relations of this algebra when these $C^{(ij)}$ are used as generators can be found in  \cite{genest2014racah}, \cite{genest2014generic}, \cite{CFGPaRL21}. Hence, $\mathfrak{R}_1$ is the symmetry algebra of the generic superintegrable model with Hamiltonian $\mathcal{H}$.

It will prove convenient to use here as generators different expressions given by affine transformations of the intermediate Casimir elements so as to fix appropriately the spectra of the operators representing these generators. Let 
\begin{align}
    L_1=&-C^{(23)} +\frac{1}{4} (b+c+2)(b+c),\\
    L_2=&-C^{(13)} +\frac{1}{4} (a+c+2)(a+c), \\
    L_3=&-C^{(12)} +\frac{1}{4} (a+b+2)(a+b). \label{L3}
\end{align}
This is to ensure that the spectra of $L_1, L_2, L_3$ are respectively
\begin{equation}
   spec(L_1)=- k(k+b+c+1), \quad spec(L_2)= -k(k+a+c+1), \quad spec(L_3)= -k(k+a+b+1), \quad \text{with } k=0,1,\ldots
\end{equation}
Take moreover $L$ to be given as follows in terms of $C^{(123)}$:
\begin{equation}
    L=-C^{(123)}+\frac{1}{4}(a+b+c+3)(a+b+c+1),
\end{equation}
so that 
\begin{equation}
    spec(L)=-n(n+a+b+c+2) \quad \text{with} \quad n=0,1,\dots \label{specL}
\end{equation}
In light of \eqref{c123cij}, we observe that
\begin{equation}
    L=L_1+L_2+L_3.
\end{equation}

Equivalent to the choice of $C^{(123)}$, $C^{(12)}$ and $C^{(23)}$ as (abstract) generators of the symmetry algebra of $\cH$, known to be $\mathfrak{R}_1$, we shall adopt in this role the elements $L, L_1, L_3$. The commutation relations between these generators can be straightforwardly computed from their expressions in $\mathfrak{su}(1,1)^{\otimes3}$. One arrives at the following presentation of the defining relations of the (centrally extended) rank one Racah algebra with $L_1$ and $L_3$ as generators and $L$ as central element:
\begin{subequations}\label{Ract}
  \begin{align}
    [[L_1,L_3],L_1]=\; &2L_1^2 + 2\{L_1,L_3\}+ \xi L_1+ \eta_1 L_3+ \zeta_1, \label{Rac1}\\
    [L_3,[L_1,L_3]]=\; &2L_3^2 + 2\{L_1,L_3\} + \xi L_3+ \eta_2 L_1+ \zeta_2,\label{Rac2}
\end{align}  
\end{subequations}
where $\{.,.\}$ stands for the anticommutator and 
\begin{align}
    \xi=& \;-(b+c)(b+1)-(b-c)(a+1)-2L=-(a+b)(b+1)-(b-a)(c+1)-2L, \label{xiJac}\\
    \eta _1=& \;-(b+c)(b+c+2),\qquad
    \eta _2= \;-(a+b)(a+b+2), \label{eta2Jac}\\
    \zeta _1=&\;(b+c)(b+1)L, \hspace{1.5cm}
    \zeta _2=\;(a+b)(b+1)L \label{zeta2Jac}.  
\end{align}

\begin{rema}
If one introduces the following rescaling and shifts of the intermediate Casimirs:
\begin{equation}
\mathcal{S}_k=4\cC^{(ij)}-4\nu_i(\nu_i-1)-4\nu_j(\nu_j-1)-\frac{1}{2},
\end{equation}
where the indices are again such that $\epsilon_{ijk}\neq 0$, using the  differential realization of $\mathfrak{su}(1,1)$, we obtain the following expressions for $\mathcal{S}_k$
 using the identification \eqref{eq:Identification Nu abc}:
\begin{equation}
\label{eq:symmetry}
\mathcal{S}_1=\hat{\mathcal{J}}_1^2+\frac{(b^2-\frac{1}{4})x_3^2}{x_2^2}+\frac{(c^2-\frac{1}{4})x_2^2}{x_3^2},\quad \mathcal{S}_2=\hat{\mathcal{J}}_2^2+\frac{(c^2-\frac{1}{4})x_1^2}{x_3^2}+\frac{(a^2-\frac{1}{4})x_3^2}{x_1^2},\quad
\mathcal{S}_3=\hat{\mathcal{J}}_3^2+\frac{(a^2-\frac{1}{4})x_2^2}{x_1^2}+\frac{(b^2-\frac{1}{4})x_1^2}{x_2^2}
\end{equation}
which are the usual expressions of the integrals of motion of $\mathcal{H}$ taken in \cite{kalnins2007wilson}, \cite{genest2014generic}, \cite{genest2014racah}, \cite{KMP96}, \cite{KMP2000}.
The following formula showing that the Hamiltonian 
\eqref{eq:decompositionCasimir} can be decomposed in terms of the symmetry operators according to
\begin{equation}
    \cH=\mathcal{S}_1+\mathcal{S}_2+\mathcal{S}_3+a^2+b^2+c^2-\frac{3}{4},
\end{equation}
is just a rewriting of \eqref{c123cij}.    
\end{rema}

\subsection{The rank two Jacobi algebra as the dynamical algebra}

Recall that the operator $\mathcal{L}$ representing $L$ in the differential realization is linearly related to the Hamiltonian $\cH$. We observed that the operators $\mathcal{L}_1$ and $\mathcal{L}_3$ commute with $\mathcal{L}$ and realize its symmetry algebra identified as $\mathfrak{R}_1$. We shall now introduce the dynamical or spectrum-generating algebra of $\mathcal{H}$ that will have $\mathfrak{R}_1$ as a subalgebra and will connect irreducibly all the states vectors of the model. This will be done by adjoining to $\mathcal{L}$, $\mathcal{L}_1$ and $\mathcal{L}_3$, the operators $\mathcal{X}^{(1)} = x_1^2$ and $\mathcal{X}^{(3)} = x_3^2$.

We thus claim that a dynamical algebra is generated by the following set of $\mathfrak{su}(1,1)^{\otimes3}$ elements 
\begin{equation}
    \{L, L_1, L_3, X_1 \equiv X^{(1)}, X_3 \equiv X^{(3)} \},
\end{equation}
where 
\begin{equation}
  X_1=\left[2J_0+J_{+}+J_{-}\right] \otimes 1\otimes 1,  \qquad X_3 = 1 \otimes 1\otimes \left[2J_0+J_{+}+J_{-}\right].
\end{equation}
The commutation relations of these generators can be obtained from the $\mathfrak{su}(1,1)$ product and the results are given in Appendix \ref{sec:Rank2Jacobi}. The algebra thus defined is recognized to be the rank two Jacobi algebra $\mathfrak{J}_2$ introduced in \cite{crampe2025rank}.

We shall aim to construct the physical representation of $\mathfrak{J}_2$ by exploiting the $\mathfrak{su}(1,1)$ framework that has been developed. This will entail determining the actions of $L$, $L_1$, $L_3$, $X_1$, and $X_3$ on the basis vectors of irreducible submodules of the triple tensor product of $\mathfrak{su}(1,1)$ positive discrete series representations. In this way, we shall see that all state vectors of the model can be related to one another through the repeated actions of the operators $L$, $L_1$, $L_3$, $X_1$, and $X_3$, which hence generate a dynamical algebra. The next step is hence to characterize the appropriate representation bases.

\section{The physical states and a representation basis of the dynamical algebra $\mathfrak{J}_2$ \label{sec:4}}

The natural quantum numbers to describe the states of the superintegrable model under consideration are the integer $n$ and $k$ that are provided by the eigenvalues of the Hamiltonian $\mathcal{H}$ (equivalently $\mathcal{L}$) and those of one constant of motion say $\mathcal{S}_1$ (or equivalently $\mathcal{L}_1$) in the notation of the last section. In order to account for the parameters of the model, we shall denote the corresponding state vectors by $|n,k;a,b,c\rangle$. It is natural to use the $\mathfrak{su}(1,1)$ module that was introduced to proceed. Given the connections established between the physical quantities and the Casimir elements, we know that the physical states will be found by recoupling the tensor product vectors that are characterized by the eigenvalues of the three algebra elements $J_0^{(i)}$, in addition to initial representation labels $\nu _i$ provided by the individual Casimir operators $C^{(i)}$. The vectors of relevance are those that are eigenstates of $C^{(123)}$ and $C^{(23)}$. A quick count shows that a label is missing to identify new basis vectors for the $\mathfrak{su}(1,1)$-module. The appropriate additional operator to diagonalize is $X^{(123)}$ which manifestly commutes with $C^{(123)}$ and $C^{(23)}$. This will lead us to use Dirac's approach to quantum mechanics with constraints. The physical space of states with support on the two-sphere, will thus be obtained through a canonical identification with the distributional states selected by the corresponding constraint on the $\mathfrak{su}(1,1)$ module. As we shall see in the next section, this auxiliary space will also play a key role in the determination of the action of the dynamical algebra generators on the physical space.

\subsection{A continuous $\mathfrak{su}(1,1)$ representation basis}\label{sec:eigenstate}

The aim of this subsection is to provide a description of the joint eigenstates of $C^{(123)}$, $C^{(23)}$ and $X^{(123)}$ in the abstract presentation of the positive discrete series representations. 
The  eigenvalue $(\nu_{123}+n)(\nu_{123}+n-1)$ of the total Casimir operator $C^{(123)}$ is known to have degeneracies of multiplicity $n+1$ in the triple tensor product of positive discrete series representations. The use of the intermediate Casimir operator $C^{(23)}$ is classical to resolve these degeneracies and the joint eigenspaces of $C^{(123)}$ and $C^{(23)}$, with respective eigenvalues $(\nu_{123}+n)(\nu_{123}+n-1)$ and $(\nu_{23}+k)(\nu_{23}+k-1)$, are known to be positive discrete series with parameters $\nu_{123}+n$. 

As announced at the beginning of this section, we first want to obtain the joint eigenstates, denoted $\;$ $\ket{r,n,k;a,b,c}$, of $C^{(123)}$, $C^{(23)}$ and $X^{(123)}$ in the triple tensor product of positive discrete series. We will see that there is only one such eigenvector for each triple $(r,n,k)$. The vectors $\ket{r,n,k;a,b,c}$  will be described in terms of the abstract basis $\ket{\nu_1;i}\ket{\nu_2;j}\ket{\nu_3;\ell}$ using the Clebsch--Gordan coefficients. 

\subsubsection{The Clebsch--Gordan coefficients $CG_{m,i,j}^{\nu_1,\nu_2}$}

Let us start with recalling the definition of the Clebsch--Gordan coefficients $CG_{m,i,j}^{\nu_1,\nu_2}$ for the tensor product of two representations of the positive discrete series of $\mathfrak{su}(1,1)$ with parameters $\nu_i$, $i=1,2$. Explicitly, they are defined by the change of basis
\begin{align}\label{eq:DefCG}
\ket{\nu_{12}+i;m}&=\sum_{j=0}^{m+i}CG_{m,i,j}^{\nu_1,\nu_2}\ket{\nu_1;j}\ket{\nu_2;m+i-j},
\end{align}
here the vector $\ket{\nu_{12}+i;m}$ is a joint eigenvector of $J_0^{(12)}$ and $C^{(12)}$.
The Clebsch--Gordan coefficients $CG_{m,i,j}^{\nu_1,\nu_2}$ are given by 
\begin{align}
\label{eq:clebsh-gordan-dfn}
CG_{m,i,j}^{\nu_1,\nu_2}&=\frac{(2\nu_1)_i}{i!}\binom{m+i}{j} p_i(j;2\nu_1-1,2\nu_2-1,i+m),
\end{align}
 where $p_i$ denotes the Hahn polynomial; its definition and properties are recalled in Appendix \ref{sec:Hahn-pols}.

 \begin{rema}
The Clebsch--Gordan coefficients are usually defined for an orthonormal basis of the discrete series representation. 
     In terms of our $CG^{\nu_1,\nu_2}_{m,i,j} $ they are given by 
     \begin{equation}
        \left( \frac{j!i!(m+i-j)!(2\nu_1)_j(2\nu_2)_{m+i-j}}{m!(2\nu_1)_j(2\nu_{12}+2i)_m(2\nu_2)_i(2\nu_{12}+i-1)_i}\right)^{\frac{1}{2}} CG^{\nu_1,\nu_2}_{m,i,j}.
     \end{equation}
 \end{rema} 

 \subsubsection{Eigenvectors of $X$} \label{subsection:eigX}

Consider the description of the eigenvectors of the operator
\begin{equation}
X=2J_0+(J_++J_-),
\end{equation}
in one copy of a positive discrete series representation acting on $\ket{\nu;i}$ as described in Section \ref{sec:Positive discrete series}. 
The spectrum of $X$ is known to be the half-line $\R^+$\cite{klimyk1995representations}. Take 
the associated generalized eigenvectors of the form, for $r\in \R^+$,
\begin{equation}\label{eq:DefEigenvectors}
    \ket{\nu;r}=\sum_{i=0}^\infty q^{(\nu)}_i(r) \ket{\nu;i}.
\end{equation}
Imposing that $X|\nu;r\rangle=r|\nu;r\rangle$, one finds that the functions $q^{(\nu)}_i(r)$ must satisfy:
\begin{equation}\label{eq:recurrence_qn}
r q^{(\nu)}_i(r)=(i+2\nu)(i+1)q^{(\nu)}_{i+1}(r)+2(\nu+i)q^{(\nu)}_i(r)+ q^{(\nu)}_{i-1}(r).  
\end{equation}
From the recurrence relation of the Laguerre polynomials \eqref{eq:LaguerreRecurrence}, we see that 
\begin{equation} \label{fnq}
     q^{(\nu)}_i(r)=\frac{(-1)^i}{(2\nu)_i }f^{(\nu)}(r)L_i^{(2\nu-1)}(r),
\end{equation}
where $f^{(\nu)}(r)$ is a function that will be determined by the normalization.

Notice that, for each $r\in \R^+$, there is only one generalized eigenvector for $X$.

Regarding the orthogonality and norm of the generalized eigenvectors $\ket{\nu;r}$ thus defined, bearing in mind \eqref{orthbasis}, we have
\begin{align}
    \langle \nu;s\ket{\nu;r}=\sum_{i,j=0}^\infty q_i^{(\nu)}(r)q_j^{(\nu)}(s)\langle\nu;i\ket{\nu;j}=f^{(\nu)}(r)f^{(\nu)}(s)\sum_{i=0}^\infty L_i^{(2\nu-1)}(r)L_i^{(2\nu-1)}(s)\frac{i!}{(2\nu)_i}.
\end{align}
This sum is infinite and has to be understood as a distribution kernel on the space $L^2(\R^+,e^{-r}r^{2\nu-1}~dr)$.
From the resolution of the identity formula \eqref{resid} of the Laguerre polynomials, we find that
\begin{equation}
    \langle \nu;s\ket{\nu;r}=\frac{\left[f^{(\nu)}(r)\right]^2}{e^{-r}r^{2\nu-1}}\Gamma(2\nu) \delta(r-s).
\end{equation}
By choosing 
\begin{equation}
    f^{(\nu)}(r)=\frac{e^{-\frac{r}{2}}r^{\nu-\frac{1}{2}}}{\Gamma(2\nu)^{\frac{1}{2}}},\label{fnf}
\end{equation}
we ensure that $\bra {\nu;s}\ket{\nu;r}=\delta(r-s).$

\subsubsection{The generalized vectors $\ket{r,n,k;a,b,c}$}

These vectors are given in the following proposition.

\begin{prop}\label{prop:vectorR,n,k}
The joint eigenvectors $\ket{r,n,k;a,b,c}$ of $C^{(123)}$, $C^{(23)}$ and $X^{(123)}$ with respective eigenvalues 
$(\nu_{123}+n)(\nu_{123}+n-1)$, $(\nu_{23}+k)(\nu_{23}+k-1)$ and $r$ have the expression
    \begin{multline}
    \ket{r,n,k;a,b,c}= \\
    \frac{1}{c_{n,k,r}}\sum_{m= 0}^{\infty} \sum_{j=0}^{m+n-k} \sum_{\ell=0}^{k+j}q_m^{(\nu_{123}+n)}(r) CG_{m,n-k,m+n-k-j}^{\nu_1,\nu_{23}+k} CG_{j,k,k+j-\ell}^{\nu_2,\nu_3}\ket{\nu_1;m+n-k-j}\ket{\nu_2;k+j-\ell} \ket{\nu_3;\ell}, \label{eq:vijtripleTensor}
\end{multline}
 where $c_{n,k,r}$ is a normalization constant to be determined. We recall that 
 \begin{equation}
  \nu_1=\frac{a+1}{2},\qquad
  \nu_2=\frac{b+1}{2},\qquad
  \nu_3=\frac{c+1}{2},\qquad
  \nu_{23}=\nu_2+\nu_3,\qquad
  \nu_{123}=\nu_1+\nu_2+\nu_3. \label{nuabc}
\end{equation}

\end{prop}

\begin{proof}
 
 Consider the vectors $\ket{\nu_{123}+n,\nu_{23}+k;m}$ that are the joint eigenvectors of $C^{(123)}$, $C^{(23)}$ and $J_0^{(123)}$ with respective eigenvalues $(\nu_{123}+n)(\nu_{123}+n-1)$, $(\nu_{23}+k)(\nu_{23}+k-1)$ and $\nu_{123}+n+m$. These can be obtained using the Clebsch--Gordan coefficients and read
\begin{equation}\label{eq:JointC123 and C23}
    \ket{\nu_{123}+n,\nu_{23}+k;m}= \sum_{i=0}^{m+n-k} \sum_{\ell=0}^{i+k} CG_{m,n-k,m+n-k-i}^{\nu_1,\nu_{23}+k} CG_{i,k,i+k-\ell}^{\nu_2,\nu_3}\ket{\nu_1;m+n-k-i}\ket{\nu_2;i+k-\ell} \ket{\nu_3;\ell}.
\end{equation} 
As already mentioned, the joint eigenspaces for $C^{(123)}$ and  $C^{(23)}$ with respective eigenvalues $(\nu_{123}+n)(\nu_{123}+n-1)$ and $(\nu_{23}+k)(\nu_{23}+k-1)$ are known to realize a positive discrete series representation with parameter $\nu_{123}+n$. Thus, making use of the description given in subsection \eqref{eq:DefEigenvectors} of the eigenvectors $\ket{\nu;r}$ of $X$, it is readily seen that the vectors $\ket{r,n,k;a,b,c}$ are expressed as follows:
\begin{equation}\label{eq:DefVij}
\ket{r,n,k;a,b,c}= \frac{1}{c_{n,k,r}}\sum_{m=0}^\infty q_m^{(\nu_{123}+n)}(r)\ket{\nu_{123}+n,\nu_{23}+k;m}.
\end{equation}
where we have introduced a normalization constant $c_{n,k,r}$.
Relations \eqref{eq:DefVij} and \eqref{eq:JointC123 and C23} lead to \eqref{eq:vijtripleTensor}.
\end{proof}

\begin{rema}
An intermediate expression of the vectors $\ket{r,n,k;a,b,c}$ will be required for future use
\begin{equation}
\ket{r,n,k;a,b,c}=  \frac{1}{c_{n,k,r}}\sum_{m\geq 0} \sum_{i=0}^{m+n-k}q_m^{(\nu_{123}+n)}(r) CG_{m,n-k,i}^{\nu_1,\nu_{23}+k} \ket{\nu_1;i} \ket{\nu_{23}+k;m+n-k-i}.\label{eq:vijDoubleTensor}
\end{equation}
It can be obtained from formula \eqref{eq:vijtripleTensor} together with the definition of the Clebsch--Gordan coefficients \eqref{eq:DefCG}.
\end{rema}

\subsubsection{Orthogonality and norms}

We conclude this subsection by providing the orthogonality properties and the distributional norms of the vectors $\ket{r,n,k;a,b,c}$
in the auxiliary Hilbert space where the positive discrete series of
$\mathfrak{su}(1,1)$ are realized.

\begin{prop}\label{prop:norm of r,n,k}
  The generalized vectors $\ket{r,n,k;a,b,c}$ form an orthogonal basis in the sense that
  \begin{equation}\label{eq:norm of r,n,k}
      \langle s,n',k';a,b,c\ket{r,n,k;a,b,c}=M_{n,k,r}
      \delta(s-r)\delta_{n,n'}\delta_{k,k'},
\end{equation}
where the coefficient $M_{n,k,r}$ is given by
\begin{align}
 M_{n,k,r}=&\frac{1}{c_{n,k,r}^2} \frac{\Gamma(2\nu_{123}+2n)}{\Gamma(2\nu_1)\Gamma(2\nu_2)\Gamma(2\nu_3)}\; \nonumber\\
 \times&\left[\frac{\Gamma(2\nu_2+k)\Gamma(2\nu_3+k)\Gamma(2\nu_1+n-k)\Gamma(2\nu_{23}+n+k)}{k!(n-k)!(2\nu_{23}+2k-1)(2\nu_{123}+2n-1)\Gamma(2\nu_{123}+n+k-1)\Gamma(2\nu_{23}+k-1)}\right].
\end{align}
\end{prop}

\begin{proof}
The orthogonality in the discrete indices $n$ and $k$ is clear since we are dealing with the eigenvalues of self-adjoint operators.
For conciseness we shall use in the proof the notations:
\begin{equation} \label{nufnuI}
    \nu_f=\nu_{123}+n \qquad \text{and} \qquad\nu_I=\nu_{23}+k. 
\end{equation}
Given the expression \eqref{eq:DefVij} for the generalized vectors we have 
    \begin{equation} \label{orthpty}
    \langle s,n',k';a,b,c\ket{r,n,k;a,b,c}=\frac{1}{c_{n,k,r}^2}\sum_{m=0}^\infty q_m^{(\nu_f)}(s)q_m^{(\nu_f)}(r)\langle\nu_f,\nu_I;m\ket{ \nu_f,\nu_I;m} \delta _{n'n}\delta_{k'k}.
    \end{equation}
Let us first determine the norm squared of $\ket{ \nu_f,\nu_I;m}$.   
    Using equation \eqref{eq:JointC123 and C23} and $\langle \nu;i\ket{\nu;i}=i!(2\nu)_i$ we obtain
    \begin{multline}
        \langle\nu_f,\nu_I;m\ket{ \nu_f,\nu_I;m}=\\  \sum_{i=0}^{m+n-k}\sum_{\ell=0}^{i+k} \left( CG_{m,n-k,m+n-k-i}^{\nu_1,\nu_I} CG_{i,k,i+k-\ell}^{\nu_2,\nu_3}\right)^2\ell!(i+k-\ell)!(m+n-k-i)!(2\nu_2)_{i+k-\ell}(2\nu_3)_\ell(2\nu_1)_{m+n-k-i}.
    \end{multline}
    Combining the definition of Clebsch--Gordan coefficients \eqref{eq:clebsh-gordan-dfn} with the orthogonality relation for Hahn polynomials \eqref{eq:orthogonality Hahn}, we get after simplification
    \begin{equation}
        \sum_{\ell=0}^{i+k}\left( CG_{i,k,i+k-\ell}^{\nu_2,\nu_3}\right)^2\ell!(i+k-\ell)!(2\nu_2)_{i+k-\ell}(2\nu_3)_{\ell}=\frac{(2\nu_3)_k(2\nu_2)_k(k+2\nu_{23}-1)_{i+k+1}i!}{(2\nu_I-1) k!}.
    \end{equation}
    A similar computation leads to
    \begin{equation}
         \langle\nu_f,\nu_I;m\ket{ \nu_f,\nu_I;m}=\frac{(2\nu_3)_k(2\nu_2)_k (2\nu_1)_{n-k}(2\nu_I)_{n-k}(2\nu_{23}+k-1)_{k+1}m!\Gamma(2\nu_f+m)}{k!(n-k)!(2\nu_I-1) (2\nu_f-1)\Gamma(2\nu_{123}+n+k-1)}.
    \end{equation}
   The insertion of this expression in \eqref{orthpty} and the use of the formulas \eqref{fnq} and \eqref{fnf} that define $q_m^{(\nu)}$ give
    \begin{align}
        \langle s,n',k';a,b,c\ket{r,n,k;a,b,c}=&\frac{1}{c_{n,k,r}^2} \frac{(2\nu_3)_k(2\nu_2)_k (2\nu_1)_{n-k}(2\nu_I)_{n-k}(2\nu_{23}+k-1)_{k+1}}{k!(n-k)!(2\nu_I-1) (2\nu_f-1)\Gamma(2\nu_{123}+n+k-1)}\nonumber\\
       \times &e^{-\frac{1}{2}(s+r)}(sr)^{(\nu_f-\frac{1}{2})}\left(\sum_{m=0}^{\infty}L_m^{(2\nu_f-1)}(s)L_m^{(2\nu_f-1)}(r)\frac{m!}{(2\nu_f)_m}\right) \delta_{n'n}\delta_{k'k}.
    \end{align}
    With the help of the resolution of the identity \eqref{resid} with the Laguerre polynomials, one finds:
    \begin{equation}
         \langle s,n',k';a,b,c\ket{r,n,k;a,b,c}=\frac{\Gamma(2\nu_f)}{c_{n,k,r}^2}\frac{(2\nu_3)_k(2\nu_2)_k (2\nu_1)_{n-k}(2\nu_I)_{n-k}(2\nu_{23}+k-1)_{k+1}}{k!(n-k)!(2\nu_I-1) (2\nu_f-1)\Gamma(2\nu_{123}+n+k-1)}\; \delta(s-r)\delta_{n'n}\delta_{k'k}.
    \end{equation}
    Making repeated use of the relation $(\rho)_m=\frac{\Gamma(\rho+m)}{\Gamma(\rho)}$ and after straightforward simplifications, one completes the proof of \ref{prop:norm of r,n,k} keeping in mind the definitions \eqref{nufnuI} of $\nu_f$ and $\nu_I$.

\end{proof}

\subsection{Dirac quantization and a discrete representation basis for $\mathfrak{J}_2$}

In order to describe the quantum model on the two-sphere $S^2$, we impose the constraint
\begin{equation}
    X^{(123)} = 1,
\end{equation}
which amounts, in the differential realization, to $x_1^2 + x_2^2 + x_3^2 = 1$.
Following Dirac’s prescription for constrained systems
\cite{Dirac2001,HenneauxTeitelboim1992},
this constraint is implemented \emph{a posteriori} on the auxiliary
$\mathfrak{su}(1,1)$ module by requiring that physical states $\ket{\psi}$ satisfy
\begin{equation}
    \bigl(X^{(123)} - 1\bigr)\ket{\psi} = 0 .
\end{equation}
Since $X^{(123)}$ has a continuous spectrum in this representation, the solutions
of this equation should be understood as generalized (distributional) vectors,
notionally supported on the two-sphere. The physical space is obtained by
identifying these constrained distributional states with the normalized vectors
\begin{equation}
    \ket{n,k;a,b,c}
    \equiv \ket{r=1,n,k;a,b,c}.
    \label{physt}
\end{equation}

It is convenient to choose for the normalization constant $c_{n,k,r}$ the following expression that only involves $n$ and shall henceforth be designated by $c_n$:
\begin{equation}\label{cn}
    c_{n,k,r} = c_n=\left[\frac{\Gamma(2\nu_{123}+2n)}{\Gamma(2\nu_1)\Gamma(2\nu_2)\Gamma(2\nu_3)}\right]^{\frac{1}{2}}.
\end{equation}

It can be checked that the Hermitian form of the
auxiliary representation, when restricted to the constrained sector, induces a
well-defined scalar product on these vectors. Mindful of \eqref{nuabc} and given the choice \eqref{cn} for $c_{n,k,r}$, this leads to following the discrete
orthogonality relation where the norm takes a factorized form, namely
\begin{equation}
    \langle n',k';a,b,c \mid n,k;a,b,c\rangle
    =
    N_{n-k}^{(a,b+c+2k+1)}\,
    N_k^{(b,c)}\,
    \delta_{n,n'}\,
    \delta_{k,k'},
    \label{scapr}
\end{equation}
with
\begin{equation}
    N_n^{(a,b)}
    =
    \frac{\Gamma(n+a+1)\Gamma(n+b+1)}
         {(2n+a+b+1)\,n!\,\Gamma(n+a+b+1)} .
    \label{normJac}
\end{equation}
On this constrained space, the radial degree of freedom is frozen and only the
discrete labels $(n,k)$ remain, yielding a discrete representation basis for the
dynamical algebra $\mathfrak{J}_2$.

\noindent In summary, we have
\begin{prop}\label{prop:norm-phys-nk}
The family of vectors $\{\ket{n,k;a,b,c}\}_{n\in\mathbb{N},\,0\le k\le n}$ defined in \eqref{physt}
forms an orthogonal basis of the physical Hilbert space obtained by Dirac
reduction from the triple tensor product of positive discrete series
representations of $\mathfrak{su}(1,1)$. Their orthogonality and norms are given by \eqref{scapr} and \eqref{normJac}.
\end{prop}

\noindent We now record the following consequence of the construction from the properties of
$\mathfrak{su}(1,1)^{\otimes 3}$ and its discrete series representations.

\begin{lemm}\label{her}
With respect to the scalar product defined by
\eqref{scapr}, the operators $L_1$, $L_3$, $X_1$ and $X_3$
are symmetric on the linear span of the basis vectors
$\{\ket{n,k;a,b,c}\}_{n\in\mathbb{N},\,0\le k\le n}$.
\end{lemm}

\begin{proof}
In the auxiliary Hilbert space, the generators
$J_0^{(i)}$ and $J_{\pm}^{(i)}$ of each copy of $\mathfrak{su}(1,1)$
are represented in such a way that
$J_0^{(i)}$ is selfadjoint and $J_+^{(i)}$ is the adjoint of $J_-^{(i)}$.
The Casimir operators $C^{(i)}$ and $C^{(ij)}$, as well as the elements
$X^{(i)}=2J_0^{(i)}+J_+^{(i)}+J_-^{(i)}$, are therefore selfadjoint.

By construction, $L$, $L_1$, $L_3$, $X_1$ and $X_3$ are real linear
combinations of selfadjoint generators in
$\mathfrak{su}(1,1)^{\otimes 3}$ and are therefore selfadjoint in the
auxiliary Hilbert space. Moreover, they commute with $X^{(123)}$, so that
the family of (generalized) states satisfying the constraint
$X^{(123)}=1$ is stable under their action. The physical scalar product is
obtained by considering the Hermitian form of the auxiliary representation
restricted to this constrained sector. With respect to this induced scalar
product, the operators $L_1$, $L_3$, $X_1$ and $X_3$ act symmetrically on
the physical vectors, which is precisely the content of the lemma.
\end{proof}

\section{The physical representation of $\mathfrak{J}_2$ \label{sec:5}}

In this section, we shall exploit the embedding of the rank two Jacobi algebra $\mathfrak{J}_2$ in $\mathfrak{su}(1,1)^{\otimes 3}$ to construct its representation in the basis provided by the  vectors $\ket{n,k;a,b,c}$ that are the eigenvectors of the generators $L$ and $L_1$. We shall first determine the action of $L_3$ thus providing the representation of the symmetry algebra already identified as the rank one Racah algebra. We shall then determine in turn the actions of the remaining two generators of $\mathfrak{J}_2$, namely $X_1$ and $X_3$. As shall be seen, to perform thse computations, it will be best to work at first in the auxiliary $\mathfrak{su}(1,1)$ module before passing to the quotient corresponding to the restriction on $S^2$. We shall also provide the raising and lowering relations for the vectors $\ket{n,k;a,b,c}$. 

\subsection{Representation of the symmetry algebra $\mathfrak{R}_1$}

The construction of the representation of the dynamical algebra $\mathfrak{J}_2$ in the basis of the physical state vectors $\ket{n,k;a,b,c}$ amounts to determining 
the action of the operators $L,L_1,L_3,X_1,X_3$ on these vectors . We shall begin with the representation of the symmetry subalgebra, i.e., the algebra generated by $L, L_1$ and $L_3$. 

Since the basis vectors $\ket{n,k;a,b,c}$ are taken to diagonalize $L$ and $L_1$ we have:
\begin{align} \label{actionLket}
    L \ket{n,k;a,b,c}&=-n(n+a+b+c+2)\ket{n,k;a,b,c},\\
    L_1\ket{n,k;a,b,c}&=-k(k+b+c+1)\ket{n,k;a,b,c}.\label{actionL1}
\end{align}
There thus remains for this first task to characterize the action of $L_3$. To that end, we will need the  Racah coefficients that is, the overlaps between the eigenvectors of $L_1$ and $L_3$ known to be affine transformations of the intermediate Casimir operators $C^{(23)}$ and $C^{(12)}$,  respectively. Let us be more precise. 
Consider again the triple tensor product of three positive discrete series representations of $\mathfrak{su}(1,1)$ with parameters $\nu_i>0,\ i=1,2,3$. The Racah coefficients are the elements of the transition matrix between the eigenbasis obtained from simultaneously diagonalizing $C^{(123)}$, $C^{(23)}$ and $J_0^{(123)}$ and the one diagonalizing jointly $C^{(123)}$, $C^{(12)}$ and $J_0^{(123)}$. The common eigenvectors of $C^{(123)}$, $C^{(23)}$ and $J_0^{(123)}$ were introduced in the proof of Proposition \ref{prop:vectorR,n,k} and read $\ket{\nu_{123}+n,\nu_{23}+k;m}$. We shall denote by $\ket{\nu_{123}+n,\nu_{12}+j;m}$ the eigenvectors of $C^{(123)}$, $C^{(12)}$ and $J_0^{(123)}$ with respective eigenvalues $(\nu_{123}+n)(\nu_{123}+n-1)$,  $(\nu_{12}+j)(\nu_{12}+j-1)$ and $m+\nu_{123}+n$. These vectors can be obtained as was done for the vectors $\ket{\nu_{123}+n,\nu_{23}+k;m}$ by recoupling the tensor product basis elements with the help of Clebsch--Gordan coefficients. One finds:
\begin{equation}
    \ket{\nu_{123}+n,\nu_{12}+j;m}=\sum_{i=0}^{m+n-j} \sum_{\ell=0}^{i+j} CG_{m,n-j,i}^{\nu_{12}+j,\nu_3} CG_{i,j,\ell}^{\nu_1,\nu_2}\ket{\nu_1;\ell}\ket{\nu_2;i+j-\ell} \ket{\nu_3;m+n-j-i}.
\end{equation}
The Racah coefficients $R_{j,k,n}^{\nu_1,\nu_2,\nu_3}$ provide the overlaps between these two bases
\begin{equation}\label{eq:DefRacahCoeff}
    \ket{\nu_{123}+n,\nu_{23}+k;m}=\sum_{j=0}^n R_{j,k,n}^{\nu_1,\nu_2,\nu_3}\ket{\nu_{123}+n,\nu_{12}+j;m}.
\end{equation}
They are explicitly given by 
\begin{equation}
  R_{j,k,n}^{\nu_1,\nu_2,\nu_3}=\binom{n}{k}\frac{(2\nu_2)_k(2\nu_1)_{n-k}(2\nu_{123}+n-1)_j}{(2\nu_1)_j(2\nu_{12}+j-1)_j(2\nu_{12}+2j)_{n-j}}R_k(\lambda(j);2\nu_2-1,2\nu_3-1,-n-1,2\nu_{12}+n-1),
\end{equation}
where $R_k$ are the Racah polynomials whose definition and properties are recalled in Appendix \ref{sec:Racah-pols}.
   
\begin{rema}
     As for the Clebsch--Gordan coefficients, the Racah coefficients are usually given for an orthonormal basis of the discrete series representation. In terms of the polynomials $R_{j,k,n}^{\nu_1,\nu_2,\nu_3}$ these coefficients read
     \begin{multline}
         \left(\frac{(n-k)!k!(2\nu_{123}+k+n-1)_{n-k}}{(n-j)!j!(2\nu_{123}+j+n-1)_{n-j}}\frac{(2\nu_1)_j(2\nu_2)_j(2\nu_3)_{n-j}(2\nu_{12}+j-1)_j(2\nu_{12}+2j)_{n-j}}{(2\nu_3)_k(2\nu_2)_k(2\nu_1)_{n-k}(2\nu_{23}+k-1)_k(2\nu_{23}+2k)_{n-k}}\right)^{\frac{1}{2}}\\
         \times \frac{(2\nu_{123}+n-1)_k}{(2\nu_{123}+n-1)_j}R_{j,k,n}^{\nu_1,\nu_2,\nu_3}.
     \end{multline}
\end{rema} 
We can now prove the following Proposition which gives the action of $L_3$ on the vectors $\ket{n,k;a,b,c}$. 
\begin{prop}\label{prop:L3Action}
    The operator $L_3$ acts as follows on the vectors $\ket{n,k;a,b,c}$
    \begin{equation}
        L_3\ket{n,k;a,b,c}=\gamma_{nk}^1 \ket{n,k+1;a,b,c}+\gamma_{nk}^2 \ket{n,k;a,b,c} +\gamma_{nk}^3 \ket{n,k-1;a,b,c},
    \end{equation}
    with
    \begin{align} \label{gamma1}
        \gamma_{nk}^1 &= \frac{(k+1)(n-k+a)(k+b+c+1)(n+k+a+b+c+2)}{(2k+b+c+1)(2k+b+c+2)},\\
        \gamma_{nk}^2 &=-\left(\frac{(n-k)(k+b+1)(k+b+c+1)(n+k+a+b+c+2)}{(2k+b+c+1)(2k+b+c+2)}\right. \label{gamma2}\\
        &\ \ +\left.\frac{k(n-k+a+1)(k+c)(n+k+b+c+1)}{(2k+b+c+1)(2k+b+c)} \right),\nonumber\\
        \gamma_{nk}^3 &=\frac{(n-k+1)(k+b)(k+c)(n+k+b+c+1)}{(2k+b+c+1)(2k+b+c)} .\label{gamma3}
    \end{align}
\end{prop}
\begin{proof}
Using definition \eqref{eq:DefVij} of the vectors $\ket{r,n,k;a,b,c}$ together with  \eqref{eq:DefRacahCoeff} to obtain the expansion of the vectors $\ket{\nu_{123}+n,\nu_{23}+k;m}$ over the basis $\{\ket{\nu_{123}+n,\nu_{12}+j;m}\}$, we have
\begin{equation}
    \ket{r,n,k;a,b,c}= \frac{1}{c_{n}}\sum_{m=0}^\infty q_m^{(\nu_{123}+n)}(r)\sum_{j=0}^{n}R_{j,k,n}^{\nu_1,\nu_2,\nu_3}\ket{\nu_{123}+n,\nu_{12}+j;m}.
\end{equation} 
Acting with the intermediate Casimir operator $C^{(12)}$, we use the fact that this operator has $\ket{\nu_{123}+n,\nu_{12}+j;m}$ as eigenvectors with eigenvalues $(\nu_{12}+j)(\nu_{12}+j-1)=j(2\nu_{12}+j-1)+\nu_{12}(\nu_{12}-1)$. We then call upon the three-term recurrence relation of the Racah polynomials \eqref{eq:Racah_Recurrence} to find that the Racah coefficients satisfy the following three-term recurrence relation
\begin{align}
j(2\nu_{12}+j-1)R_{j,k,n}^{\nu_1,\nu_2,\nu_3}
&= -\frac{(k+1)(2\nu_1+n-k-1)(2\nu_{23}+k-1)(2\nu_{123}+n+k-1)}
        {(2\nu_{23}+2k-1)(2\nu_{23}+2k)} 
   R_{j,k+1,n}^{\nu_1,\nu_2,\nu_3} \nonumber\\[4pt]
&\quad + \left(
      \frac{(n-k)(k+2\nu_2)(k+2\nu_{23}-1)(2\nu_{123}+n+k-1)}
           {(2\nu_{23}+2k-1)(2\nu_{23}+2k)}
    \right. \nonumber\\
&\qquad\quad \left.
    + \frac{k(2\nu_1+n-k)(2\nu_{3}+k-1)(2\nu_{23}+n+k-1)}
           {(2\nu_{23}+2k-1)(2\nu_{23}+2k-2)} 
    \right)
    R_{j,k,n}^{\nu_1,\nu_2,\nu_3} \nonumber\\[4pt]
&\quad - \frac{(n-k+1)(2\nu_2+k-1)(2\nu_{3}+k-1)(2\nu_{23}+n+k-1)}
              {(2\nu_{23}+2k-1)(2\nu_{23}+2k-2)}
    R_{j,k-1,n}^{\nu_1,\nu_2,\nu_3}. 
\end{align}
To conclude we
use the fact that $L_3=-C^{(12)}+\frac{1}{4}(a+b+2)(a+b)$ as per \eqref{L3}, while the link with the parameters $a,b,c$ is made through \eqref{nuabc}. Finally, a direct computation yields the result. 
\end{proof}

The next two subsections will provide the actions of the generators $X_1$ and $X_3$ thereby completing the construction of the representation of $\mathfrak{J}_2$ in the physical basis. In both cases, the transformations of the recombined state vectors will have to be recovered from the simple actions of $X_1$ or $X_3$ on the first or third factor of the tensor product state vectors.

\subsection{Action of $X_1$}

The proposition below gives the action of $X_1$ we are looking for.
\begin{prop}\label{prop:X1Action}
The operator $X_1$ acts in the following three–diagonal fashion in the basis
$\{\ket{n,k;a,b,c}\}$:
\begin{equation}
X_1 \ket{n,k;a,b,c}
=
\alpha^1_{nk} \ket{n+1,k;a,b,c}
+\alpha^2_{nk} \ket{n,k;a,b,c}
+\alpha^3_{nk} \ket{n-1,k;a,b,c},
\label{actX1}
\end{equation}
with
\begin{align}
\alpha^1_{nk}
&=
-\frac{(n-k+1)(n+k+a+b+c+2)}
{(2n+a+b+c+2)(2n+a+b+c+3)},\\
\alpha^2_{nk}
&=
-\alpha^1_{nk}-\alpha^3_{nk}
=
\frac12\left(
1-
\frac{(2k+a+b+c+1)(2k-a+b+c+1)}
{(2n+a+b+c+1)(2n+a+b+c+3)}
\right),\\
\alpha^3_{nk}
&=
-\frac{(n+k+b+c+1)(n-k+a)}
{(2n+a+b+c+1)(2n+a+b+c+2)}.
\end{align}
\end{prop}

\begin{proof}

We set again $\nu_f=\nu_{123}+n$ and bear \eqref{nuabc} in mind.
For computational reasons we first consider the action of $X_1$ on the
auxiliary vectors $\ket{r,n,k;a,b,c}$ and use also the formula
\eqref{eq:vijDoubleTensor}.  We give a direct proof: the left–hand side of
\eqref{actX1} with $r$ reinstated is computed explicitly and shown to be
equal to the right–hand side. More explicitly, we will prove the following
\begin{equation}
\label{eq:three-term-r}
X_1 \ket{r,n,k;a,b,c}
=
\alpha^1_{nk} r \ket{r,n+1,k;a,b,c}
+
\alpha^2_{nk} r \ket{r,n,k;a,b,c}
+
\alpha^3_{nk} r \ket{r,n-1,k;a,b,c}.
\end{equation}
We want to collect the coefficient in front of the vector
$\ket{\nu_1;i}\ket{\nu_{23}+k;n-k+m-i}$ on both sides.
On the left–hand side we obtain
\begin{equation}\label{eq:LHSdouble}
\frac{1}{c_n}\left(
(i+1)(i+2\nu_1)q_{m+1}^{(\nu_f)}(r)
CG_{m+1,n-k,i+1}^{\nu_1,\nu_{23}+k}
+
2(i+\nu_1)q_m^{(\nu_f)}(r)
CG_{m,n-k,i}^{\nu_1,\nu_{23}+k}
+
q_{m-1}^{(\nu_f)}(r)
CG_{m-1,n-k,i-1}^{\nu_1,\nu_{23}+k}
\right).
\end{equation}
The right–hand side gives
\begin{align}\label{eq:RHSdouble}
\frac{\alpha_{nk}^1 r}{c_{n+1}}
\,q_{m-1}^{(\nu_f+1)}(r)\,
CG_{m-1,n-k+1,i}^{\nu_1,\nu_{23}+k}+
\frac{\alpha_{nk}^2 r}{c_n}
\,q_m^{(\nu_f)}(r)\,
CG_{m,n-k,i}^{\nu_1,\nu_{23}+k}+
\frac{\alpha_{nk}^3 r}{c_{n-1}}
\,q_{m+1}^{(\nu_f-1)}(r)\,
CG_{m+1,n-k-1,i}^{\nu_1,\nu_{23}+k}.
\end{align}
Using twice the contiguity relations for Laguerre polynomials
\eqref{eq:contiguityLaguerre1}–\eqref{eq:contiguityLaguerre2},
we obtain the following relations for the normalized functions
$q_m^{(\nu_f)}$:
\begin{align}
r\,q_{m-1}^{(\nu_f+1)}(r)
&=
\sqrt{2\nu_f(2\nu_f+1)}
\left(
m(m+1)q_{m+1}^{(\nu_f)}(r)
+2m q_m^{(\nu_f)}(r)
+q_{m-1}^{(\nu_f)}(r)
\right),
\label{eq:contiguity-q1-new}
\\
r\,q_{m+1}^{(\nu_f-1)}(r)
&=
\frac{1}{\sqrt{(2\nu_f-2)(2\nu_f-1)}}
\left(
(2\nu_f+m-1)(2\nu_f+m)q_{m+1}^{(\nu_f)}(r)
+
2(2\nu_f+m-1)q_m^{(\nu_f)}(r)
+
q_{m-1}^{(\nu_f)}(r)
\right).
\label{eq:contiguity-q2-new}
\end{align}
Because of the choice \eqref{cn} of the normalization constants $c_n$,
the square–root factors are exactly compensated by the ratios
$c_{n+1}/c_n$ and $c_{n-1}/c_n$.
Using these two relations together with the recurrence relation
\eqref{eq:recurrence_qn}, we can arrange the formulas so as to only have
functions $q_m^{(\nu_f)}$ with the same parameter $\nu_f$ appearing.
Expanding \eqref{eq:RHSdouble} in the basis
$q_{m+1}^{(\nu_f)}, q_m^{(\nu_f)}, q_{m-1}^{(\nu_f)}$
and identifying the coefficients with those of \eqref{eq:LHSdouble},
we obtain respectively the following three contiguity relations for the
Clebsch–Gordan coefficients:
\begin{align}
(i+1)(i+2\nu_1)
CG_{m+1,n-k,i+1}^{\nu_1,\nu_{23}+k}
&=
m(m+1)\alpha_{nk}^1
CG_{m-1,n-k+1,i}^{\nu_1,\nu_{23}+k}
+
(m+1)(m+2\nu_f)\alpha_{nk}^2
CG_{m,n-k,i}^{\nu_1,\nu_{23}+k}
\nonumber\\
&\qquad
+
(m+2\nu_f-1)(m+2\nu_f)\alpha_{nk}^3
CG_{m+1,n-k-1,i}^{\nu_1,\nu_{23}+k},
\\
(i+\nu_1)
CG_{m,n-k,i}^{\nu_1,\nu_{23}+k}
&=
m\alpha_{nk}^1
CG_{m-1,n-k+1,i}^{\nu_1,\nu_{23}+k}
+
(m+\nu_f)\alpha_{nk}^2
CG_{m,n-k,i}^{\nu_1,\nu_{23}+k}
\nonumber\\
&\qquad
+
(2\nu_f+m-1)\alpha_{nk}^3
CG_{m+1,n-k-1,i}^{\nu_1,\nu_{23}+k},
\\
CG_{m-1,n-k,i-1}^{\nu_1,\nu_{23}+k}
&=
\alpha_{nk}^1
CG_{m-1,n-k+1,i}^{\nu_1,\nu_{23}+k}
+
\alpha_{nk}^2
CG_{m,n-k,i}^{\nu_1,\nu_{23}+k}
+
\alpha^3_{nk}
CG_{m+1,n-k-1,i}^{\nu_1,\nu_{23}+k}.
\end{align}
To conclude the proof one checks that the first relation coincides with
the contiguity relation (HII/I) for the Hahn polynomials,
the second one with the recurrence relation \eqref{eq:recuR},
and the third with the contiguity relation (HI/II).
The corresponding identities are listed in Appendix
\ref{sec:Hahn-pols}.  The last step is obtained by setting $r=1$
in \eqref{eq:three-term-r}.
\end{proof}

\subsection{Action of $X_3$}

We now complete the construction of the physical representation of $\mathfrak{J}_2$ by determining the action of $X_3$ on the vectors $\ket{n,k;a,b,c}$. The result is given in the  proposition below.
\begin{prop}\label{prop:X3Action}
The operator $X_3$ has the following nine-term action in the basis $\{\ket{n,k;a,b,c}\}$
\begin{multline}
\label{eq:9term}
X_3 \ket{n,k;a,b,c}=\beta^1_{nk}\ket{n+1,k-1;a,b,c}+\beta^2_{nk} \ket{n+1,k;a,b,c}+\beta^3_{nk} \ket{n+1,k+1;a,b,c}\\
+\beta^4_{nk} \ket{n,k-1;a,b,c}+  \beta^5_{nk} \ket{n,k;a,b,c}+\beta^6_{nk} \ket{n,k+1;a,b,c}
\\+\beta^7_{nk} \ket{n-1,k-1;a,b,c}+\beta^8_{nk} \ket{n-1,k;a,b,c}+\beta^9_{nk} \ket{n-1,k+1;a,b,c},
\end{multline}
with
\begin{align}
\beta_{nk}^1&=\frac{(k+b)(k+c)(n-k+1)(n-k+2)}{(2k+b+c)(2k+b+c+1)(2n+a+b+c+2)(2n+a+b+c+3)},\\
\beta_{nk}^2&=\frac{(n-k +1) (n+k+a+b+c+2)}{2 (2n+a+b+c+2)(2n+a+b+c+3)}\left(\frac{c^{2}-b^{2}}{(2k +b+c)(2k+b+c+2)}+1\right),\\
\beta_{nk}^3&=\frac{(k+1)(k+b+c+1)(n+k+a+b+c+2)(n+k+a+b+c+3)}{(2k+b+c+1)(2k+b+c+2)(2n+a+b+c+2)(2n+a+b+c+3)},\\
\beta_{nk}^4&=\frac{2(k+b)(k+c)(n-k+1)(n+k+b+c+1)}{(2k+b+c)(2k+b+c+1)(2n+a+b+c+1)(2n+a+b+c+3)},\\
\beta_{nk}^5&=\frac{3}{4}+\frac{(2k+a+b+c+1)(2k-a+b+c+1)}{4 (2n+a+b+c+1)(2n+a+b+c+3)}\nonumber\\
&\ \ + \frac{c^2-b^2}{4}\left(\frac{1}{(2k+b+c)(2k+b+c+2)} +\frac{1}{(2n+a+b+c+1)(2n+a+b+c+3)}\right.\nonumber\\
&\ \ \left.+\frac{1-a^2}{(2k +b+c)(2k+b+c+2) (2n+a+b+c+1)(2n+a+b+c+3)}\right),\\
\displaybreak[1]
\beta_{nk}^6&=\frac{2(k+1)(k+b+c+1)(n-k+a)(n+k+a+b+c+2)}{(2k+b+c+1)(2k+b+c+2)(2n+a+b+c+1)(2n+a+b+c+3)},\\
\beta^7_{nk}&=\frac{(k+b)(k+c)(n+k+b+c)(n+k+b+c+1)}{(2k+b+c)(2k+b+c+1)(2n+a+b+c+1)(2n+a+b+c+2)},\\
\beta_{nk}^8&=\frac{(n-k+a) (n+k+b+c+1)}{2 (2n+a+b+c+1) (2n+a+b+c+2)}\left(\frac{c^2-b^2}{(2k +b+c) (2k+b+c+2)}+1\right),\\
\beta_{nk}^9&=\frac{(k+1)(k+b+c+1)(n-k+a-1)(n-k+a)}{(2k+b+c+1)(2k+b+c+2)(2n+a+b+c+1)(2n+a+b+c+2)}.
\end{align}
\end{prop}
\begin{proof}
    We again write $\nu_f=\nu_{123}+n$ and recall \eqref{nuabc}. As was done for Proposition \ref{prop:X1Action}, we shall prove an extension of equation \eqref{eq:9term} to the vectors $\ket{r,n,k;a,b,c}$ which reads:
        \begin{multline}
        X_3 \ket{r,n,k;a,b,c}=\beta^1_{nk}r \ket{r,n+1,k-1;a,b,c}+\beta^2_{nk}r \ket{r,n+1,k;a,b,c}+\beta^3_{nk}r \ket{r,n+1,k+1;a,b,c}\\
        +\beta^4_{nk}r \ket{r,n,k-1;a,b,c}+  \beta^5_{nk}r \ket{r,n,k;a,b,c}+\beta^6_{nk} r\ket{r,n,k+1;a,b,c}
        \\+\beta^7_{nk}r \ket{r,n-1,k-1;a,b,c}+\beta^8_{nk}r \ket{r,n-1,k;a,b,c}+\beta^9_{nk}r \ket{r,n-1,k+1;a,b,c}. \label{9termext}
    \end{multline}
    After this reinsertion of $r$, we shall show that the left hand side that is found upon computing the action of $X_3$ matches the 9 terms given on the right hand side. We shall use formula \eqref{eq:vijtripleTensor}.
     We start by looking at the coefficient in front of $\ket{\nu_1;m+n-k-j}\ket{\nu_2;j+k-\ell} \ket{\nu_3;\ell}$  in both sides of equation \eqref{9termext}. 
    On the left hand side, we find
    \begin{multline}
    \label{eq:9terms-lhs}
       \frac{1}{c_n}\left((\ell+1)(2\nu_3+\ell)q_{m+1}^{(\nu_f)}(r)CG_{m+1,n-k,m+n-k-j}^{\nu_1,\nu_{23}+k}CG_{j+1,k,k+j-\ell}^{\nu_2,\nu_3}\right.\\
      \left. +2(\nu_3+\ell)q_m^{(\nu_f)}(r)CG_{m,n-k,m+n-k-j}^{\nu_1,\nu_{23}+k}CG_{j,k,k+j-\ell}^{\nu_2,\nu_3}+q_{m-1}^{(\nu_f)}(r)CG_{m-1,n-k,m+n-k-j}^{\nu_1,\nu_{23}+k}CG_{j-1,k,k+j-\ell}^{\nu_2,\nu_3} \right).
    \end{multline}
On the right hand side, we get
\begin{align}
\label{eq:9terms-rhs}
&\frac{\beta_{nk}^1}{c_{n+1}}r\,q^{(\nu_f+1)}_{m-1}(r)CG_{m-1,n-k+2,m+n-k-j}^{\nu_1,\nu_{23}+k-1}CG_{j+1,k-1,k+j-\ell}^{\nu_2,\nu_3}+\frac{\beta_{nk}^2}{c_{n+1}}r\,q^{(\nu_f+1)}_{m-1}(r)CG_{m-1,n-k+1,m+n-k-j}^{\nu_1,\nu_{23}+k}CG_{j,k,k+j-\ell}^{\nu_2,\nu_3}\nonumber\\
&+\frac{\beta_{nk}^3}{c_{n+1}}r\,q^{(\nu_f+1)}_{m-1}(r)CG_{m-1,n-k,m+n-k-j}^{\nu_1,\nu_{23}+k+1}CG_{j-1,k+1,k+j-\ell}^{\nu_2,\nu_3}+\frac{\beta_{nk}^4}{c_{n}}r\,q^{(\nu_f)}_{m}(r)CG_{m,n-k+1,m+n-k-j}^{\nu_1,\nu_{23}+k-1}CG_{j+1,k-1,k+j-\ell}^{\nu_2,\nu_3}\nonumber\\
&+\frac{\beta_{nk}^5}{c_{n}}r\,q^{(\nu_f)}_{m}(r)CG_{m,n-k,m+n-k-j}^{\nu_1,\nu_{23}+k}CG_{j,k,k+j-\ell}^{\nu_2,\nu_3}+\frac{\beta_{nk}^6}{c_{n}}r\,q^{(\nu_f)}_{m}(r)CG_{m,n-k-1,m+n-k-j}^{\nu_1,\nu_{23}+k+1}CG_{j-1,k+1,k+j-\ell}^{\nu_2,\nu_3}\nonumber\\
&+\frac{\beta_{nk}^7}{c_{n-1}}r\,q^{(\nu_f-1)}_{m+1}(r)CG_{m+1,n-k,m+n-k-j}^{\nu_1,\nu_{23}+k-1}CG_{j+1,k-1,k+j-\ell}^{\nu_2,\nu_3}+\frac{\beta_{nk}^8}{c_{n-1}}r\,q^{(\nu_f-1)}_{m+1}(r)CG_{m+1,n-k-1,m+n-k-j}^{\nu_1,\nu_{23}+k}CG_{j,k,k+j-\ell}^{\nu_2,\nu_3}\nonumber\\
&+\frac{\beta_{nk}^9}{c_{n-1}}r\,q^{(\nu_f-1)}_{m+1}(r)CG_{m+1,n-k-2,m+n-k-j}^{\nu_1,\nu_{23}+k+1}CG_{j-1,k+1,k+j-\ell}^{\nu_2,\nu_3}.
\end{align}
We again apply the contiguity relations \eqref{eq:contiguity-q1-new} and \eqref{eq:contiguity-q2-new} as well as the recurrence relation \eqref{eq:recurrence_qn} of the Laguerre functions $q_m^{(\nu_f)}$ so as to arrange that only Laguerre functions with the same parameter $\nu_f$ remain in the formulas. 
We can now expand \eqref{eq:9terms-rhs} into the basis consisting of these Laguerre functions $q_m^{(\nu_f)}$. By identifying the coefficients of $q_{m+1}^{(\nu_f)}$, in \eqref{eq:9terms-lhs} and \eqref{eq:9terms-rhs}, we obtain the following equation involving the Clebsch--Gordan coefficients 
\begin{align}
\label{eq:coef-pm1-9terms}
&CG_{j+1,k-1,k+j-\ell}^{\nu_2,\nu_3}\left(m(m+1)\beta_{nk}^1CG_{m-1,n-k+2,m+n-k-j}^{\nu_1,\nu_{23}+k-1}+(m+1)(2\nu_f+m)\beta_{nk}^4CG^{\nu_1,\nu_{23}+k-1}_{m,n-k+1,m+n-k-j}\right.\nonumber\\
&\left.\hspace{2.7cm}+(2\nu_f+m)(2\nu_f+m-1)\beta_{nk}^7CG_{m+1,n-k,m+n-k-j}^{\nu_1,\nu_{23}+k-1} \right)\nonumber\\
&+CG_{j,k,k+j-\ell}^{\nu_2,\nu_3}\left(m(m+1)\beta_{nk}^2CG_{m-1,n-k+1,m+n-k-j}^{\nu_1,\nu_{23}+k}+(m+1)(2\nu_f+m)\beta_{nk}^5CG_{m,n-k,m+n-k-j}^{\nu_1,\nu_{23}+k}\right.\nonumber\\
&\left. \hspace{2.5cm}+(2\nu_f+m)(2\nu_f+m-1)\beta_{nk}^8CG_{m+1,n-k-1,m+n-k-j}^{\nu_1,\nu_{23}+k} \right)\nonumber\\
&+CG_{j-1,k+1,k+j-\ell}^{\nu_2,\nu_3}\left(m(m+1)\beta_{nk}^3CG_{m-1,n-k,m+n-k-j}^{\nu_1,\nu_{23}+k+1} +(m+1)(2\nu_f+m)\beta_{nk}^6CG_{m,n-k-1,m+n-k-j}^{\nu_1,\nu_{23}+k+1}\right.\nonumber\\
&\left.\hspace{3.2cm}+(2\nu_f+m)(2\nu_f+m-1)\beta_{nk}^9CG_{m+1,n-k-2,m+n-k-j}^{\nu_1,\nu_{23}+k+1}\right)\nonumber\\
&=(\ell+1)(2\nu_3+\ell)CG_{m+1,n-k,m+n-k-j}^{\nu_1,\nu_{23}+k}CG_{j+1,k,k+j-\ell}^{\nu_2,\nu_3}.
\end{align}
Similarly by identifying the coefficients of $q_{m}^{(\nu_f)}$, in \eqref{eq:9terms-lhs} and \eqref{eq:9terms-rhs}, we find
\begin{align}
\label{eq:coef-pm0-9terms}
&CG_{j+1,k-1,k+j-\ell}^{\nu_2,\nu_3}\left(m\beta_{nk}^1CG_{m-1,n-k+2,m+n-k-j}^{\nu_1,\nu_{23}+k-1}+(\nu_f+m)\beta_{nk}^4CG_{m,n-k+1,m+n-k-j}^{\nu_1,\nu_{23}+k-1}\right.\nonumber \\
&\left. \hspace{2.7cm}+(2\nu_f+m-1)\beta_{nk}^7CG_{m+1,n-k,m+n-k-j}^{\nu_1,\nu_{23}+k-1} \right)\nonumber\\
&+CG_{j,k,k+j-\ell}^{\nu_2,\nu_3}\left(m\beta_{nk}^2CG_{m-1,n-k+1,m+n-k-j}^{\nu_1,\nu_{23}+k}+(m+\nu_f)\beta_{nk}^5CG_{m,n-k,m+n-k-j}^{\nu_1,\nu_{23}+k}\right.\nonumber\\
&\left.\hspace{2.5cm}+(2\nu_f+m-1)\beta_{nk}^8CG_{m+1,n-k-1,m+n-k-j}^{\nu_1,\nu_{23}+k} \right)\nonumber\\
&+CG_{j-1,k+1,k+j-\ell}^{\nu_2,\nu_3}\left(m\beta_{nk}^3CG_{m-1,n-k,m+n-k-j}^{\nu_1,\nu_{23}+k+1}+(\nu_f+m)\beta_{nk}^6CG_{m,n-k-1,m+n-k-j}^{\nu_1,\nu_{23}+k+1}\right.\nonumber\\
&\left.\hspace{3.2cm}+(2\nu_f+m-1)\beta_{nk}^9CG_{m+1,n-k-2,m+n-k-j}^{\nu_1,\nu_{23}+k+1}\right)\nonumber\\
&=(\nu_3+\ell)CG_{m,n-k,m+n-k-j}^{\nu_1,\nu_{23}+k}CG_{j,k,k+j-\ell}^{\nu_2,\nu_3}.
\end{align}
Finally by identifying the coefficients of $q_{m-1}^{(\nu_f)}$, in \eqref{eq:9terms-lhs} and \eqref{eq:9terms-rhs}, we get
\begin{align}
\label{eq:coef-pmm1-9terms}
&CG_{j+1,k-1,k+j-\ell}^{\nu_2,\nu_3}\left(\beta_{nk}^1CG_{m-1,n-k+2,m+n-k-j}^{\nu_1,\nu_{23}+k-1}+\beta_{nk}^4CG_{m,n-k+1,m+n-k-j}^{\nu_1,\nu_{23}+k-1}+\beta_{nk}^7CG_{m+1,n-k,m+n-k-j}^{\nu_1,\nu_{23}+k-1} \right)\nonumber\\
&+CG_{j,k,k+j-\ell}^{\nu_2,\nu_3}\left(\beta_{nk}^2CG_{m-1,n-k+1,m+n-k-j}^{\nu_1,\nu_{23}+k}+\beta_{nk}^5CG_{m,n-k,m+n-k-j}^{\nu_1,\nu_{23}+k}+\beta_{nk}^8CG_{m+1,n-k-1,m+n-k-j}^{\nu_1,\nu_{23}+k} \right)\nonumber\\
&+CG_{j-1,k+1,k+j-\ell}^{\nu_2,\nu_3}\left(\beta_{nk}^3CG_{m-1,n-k,m+n-k-j}^{\nu_1,\nu_{23}+k+1}+\beta_{nk}^6CG_{m,n-k-1,m+n-k-j}^{\nu_1,\nu_{23}+k+1}+\beta_{nk}^9CG_{m+1,n-k-2,m+n-k-j}^{\nu_1,\nu_{23}+k+1}\right)\nonumber\\
&=CG_{m-1,n-k,m+n-k-j}^{\nu_1,\nu_{23}+k}CG_{j-1,k,k+j-\ell}^{\nu_2,\nu_3}.
\end{align}
We shall now demonstrate the validity of equation \eqref{eq:coef-pm1-9terms}; the proofs of equations \eqref{eq:coef-pm0-9terms} and \eqref{eq:coef-pmm1-9terms} follow analogously. 

Applying equation \eqref{eq:clebsh-gordan-dfn} and the $B_2$-contiguity relation (HIV/III) given in Appendix \ref{sec:Hahn-pols} with the parameters $\alpha=2\nu_2-1,$ $\beta=2\nu_3-1,$ $N=k+j+1$, $x=k+j-\ell,$ the right hand side of equation \eqref{eq:coef-pm1-9terms} becomes
\begin{equation}
\label{eq:contiguity-HIV/III-clebsch}
CG_{m+1,n-k,m+n-k-j}^{\nu_1,\nu_{23}+k}\left(\Psi_{k,j}^{+1,+} CG_{j-1,k+1,k+j-\ell}^{\nu_2,\nu_3}+\Psi_{k,j}^{0,+} CG_{j,k,k+j-\ell}^{\nu_2,\nu_3}+\Psi_{k,j}^{-1,+} CG_{j+1,k-1,k+j-\ell}^{\nu_2,\nu_3} \right).
\end{equation}
with 
\begin{align}
\Psi_{k,j}^{+1,+}&=\frac{j\left(j +1\right) \left(k +2 \nu_{23}-1 \right)(k+1)}{ \left(2 k+2 \nu_{23}-1 \right)\left(2k +2\nu_{23} \right)},\\
\Psi_{k,j}^{0,+}&=\frac{\left(j +1\right)\left(k^2+2 k \nu_2 +2 k \nu_3 +2 \nu_3 \nu_2 +2 \nu_3^{2}-k -2 \nu_3 \right)  \left(2 k +2 \nu_{23} +j \right)}{2 \left(k+\nu_{23}-1 \right) \left(k +\nu_{23} \right) },\\
\Psi_{k,j}^{-1,+}&=\frac{\left(2 k +2 \nu_{23} +j \right) \left(2 k+2 \nu_{23} +j-1 \right) \left(2 \nu_3+k-1 \right)(2\nu_2+k-1)}{ \left(2 k+2 \nu_{23}-1 \right)\left(2k+2\nu_{23}-2\right)}.
\end{align}
Identifying the coefficients of $CG_{j-1,k+1,k+j-\ell}^{\nu_2,\nu_3}, CG_{j,k,k+j-\ell}^{\nu_2,\nu_3}$ and $CG_{j+1,k-1,k+j-\ell}^{\nu_2,\nu_3}$, in the left hand side of equation \eqref{eq:coef-pm1-9terms} and equation \eqref{eq:contiguity-HIV/III-clebsch}, one gets the following three contiguity relations for the Clebsch--Gordan coefficients
\begin{multline}
\Psi_{n-k,j}^{-1,+}CG_{m+1,n-k,m+n-k-j}^{\nu_1,\nu_{23}+k}=m(m+1)\beta_{nk}^1CG_{m-1,n-k+2,m+n-k-j}^{\nu_1,\nu_{23}+k-1}\\+(m+1)(2\nu_f+m)\beta_{nk}^4CG^{\nu_1,\nu_{23}+k-1}_{m,n-k+1,m+n-k-j}+(2\nu_f+m)(2\nu_f+m-1)\beta_{nk}^7CG_{m+1,n-k,m+n-k-j}^{\nu_1,\nu_{23}+k-1},
\end{multline}
\begin{multline}
\Psi_{n-k,j}^{0,+}CG_{m+1,n-k,m+n-k-j}^{\nu_1,\nu_{23}+k}=m(m+1)\beta_{nk}^2CG_{m-1,n-k+1,m+n-k-j}^{\nu_1,\nu_{23}+k}+(m+1)(2\nu_f+m)\beta_{nk}^5CG_{m,n-k,m+n-k-j}^{\nu_1,\nu_{23}+k}\\ +(2\nu_f+m)(2\nu_f+m-1)\beta_{nk}^8CG_{m+1,n-k-1,m+n-k-j}^{\nu_1,\nu_{23}+k},
\end{multline}
\begin{multline}
\Psi_{n-k,j}^{+1,+}CG_{m+1,n-k,m+n-k-j}^{\nu_1,\nu_{23}+k}=m(m+1)\beta_{nk}^3CG_{m-1,n-k,m+n-k-j}^{\nu_1,\nu_{23}+k+1} +(m+1)(2\nu_f+m)\beta_{nk}^6CG_{m,n-k-1,m+n-k-j}^{\nu_1,\nu_{23}+k+1}\\+(2\nu_f+m)(2\nu_f+m-1)\beta_{nk}^9CG_{m+1,n-k-2,m+n-k-j}^{\nu_1,\nu_{23}+k+1}.
\end{multline}
To complete the proof one checks that the first equation is equivalent to the $B_2'$-contiguity relation (HIV/IV) for the Hahn polynomials, that the second one is equivalent to the $B_2$-contiguity relation (HIV/III)  and finally that the last one is the $B_2'$-contiguity relation (HIII/III) with all these contiguity relation given in Appendix \ref{sec:Hahn-pols}. The final step of the proof is achieved by reverting to $r=1$.
\end{proof}

\subsection{Raising and lowering relations}
We end this section with the description of the lowering and raising relations acting on the state vectors \,      $\ket{n,k;a,b,c}$ using the generators of the rank 2 Jacobi algebra. To do so we follow the ideas used for the Askey--Wilson polynomials in \cite{Koo2007}. The raising and lowering relations for the parameter $n$ can be described as follows
\begin{align} 
       & \big((2n+a+b+c+1)[X_1,L]+4[[X_1,L],L]\big)\ket{n,k;a,b,c}\nonumber\\
       &\hspace{3cm} =2(k-n-1)(n+k+a+b+c+2)\ket{n+1,k;a,b,c},\\
      &  \big((2n+a+b+c+3)[X_1,L]-4[[X_1,L],L]\big)\ket{n,k;a,b,c}\nonumber\\
      &\hspace{3cm}=2(n+k+b+c+1)(n-k+a)\ket{n-1,k;a,b,c}.
\end{align}
The first relation is a consequence of  
    \begin{equation}
    (\lambda_{n}-\lambda_{n-1})[X_1,L]\ket{n,k;a,b,c}=\alpha^1_{nk}(\lambda_{n}-\lambda_{n-1})(\lambda_{n}-\lambda_{n+1})\ket{n+1,k;a,b,c}+\alpha^3_{nk}(\lambda_{n}-\lambda_{n-1})^2\ket{n-1,k;a,b,c},
    \end{equation}
    and 
    \begin{equation}
    [[X_1,L],L]\ket{n,k;a,b,c}=\alpha^1_{nk}(\lambda_{n}-\lambda_{n-1})^2\ket{n+1,k;a,b,c}+\alpha^3_{nk}(\lambda_{n}-\lambda_{n-1})^2\ket{n-1,k;a,b,c},
    \end{equation}
    where $\lambda _n = -n(n+a+b+c+2)$ is the eigenvalue of $L$ given in formula \eqref{specL} and $\alpha^1_{nk},\ \alpha^3_{nk}$ are defined in Proposition \ref{prop:X1Action}. The difference between both equations leads to the formula for the raising relation after a direct computation of the coefficients. The formula for the lowering relation is obtained in a similar fashion.

Similarly, the raising and lowering relations for the parameter $k$ are given by
    \begin{align}
      &  \big( (b+c+2k)[L_3,L_1]+4[[L_3,L_1],L_1]\big)\ket{n,k;a,b,c}\nonumber\\
       & \hspace{3cm}=2(k+1)(n-k+a)(k+b+c+1)(n+k+a+b+c+2)\ket{n,k+1;a,b,c},\\
       & \big(-(b+c+2k+2)[L_3,L_1]+4[[L_3,L_1],L_1]\big)\ket{n,k;a,b,c}\nonumber\\
       &\hspace{3cm} =2(n-k+1)(k+b)(k+c)(n+k+b+c+1)\ket{n,k-1;a,b,c}.
    \end{align}
The proof is similar to the raising and lowering relations for the parameter $n$.


\section{Algebraic solution of the generic superintegrable model on $S^2$}\label{sec:eigenstateH}

In this section we wish to explain how the model on $S^2$ with Hamiltonian \eqref{eq:Identification Nu abc} can be solved algebraically by exploiting its dynamical algebra $\mathfrak{J}_2$. The goal is to obtain the wavefunctions of this system from the representations of $\mathfrak{J}_2$.

\subsection{The position eigenbasis in barycentric coordinates}

The operators $X_i, i=1,2,3$ are naturally identified as position operators. We shall take the joint eigenvectors of $X_1, X_2,X_3$ to have for eigenvalues $x,y,z$. We already observed at the beginning of subsection 4.1.2 that their range is $(0,+\infty)$. Since $X^{(123)}=X_1+X_2+X_3 \sim 1 $ on the physical space we have $x+y+z=1$. A simple argument (see the remark below)
shows that this implies that $0\leq x\leq 1-y\leq 1$. These eigenvalues can thus be identified to the barycentric coordinates of the triangle that we can view as an octant of $S^2$. We shall therefore use the following basis vectors:
\begin{equation} \label{posbas}
    |x,y;a,b,c\rangle =|\nu_1;x\rangle|\nu_2;y\rangle|\nu_3;z=1-x-y\rangle \qquad  \text{with} \qquad 0\leq x\leq 1-y\leq 1,
\end{equation}
in the notation of subsection \ref{subsection:eigX} and with \eqref{nuabc} providing the relations between $a, b, c$ and $\nu_1, \nu_2, \nu_3$. Given the normalisation of the vectors $\ket{\nu,x}$ we have
\begin{equation}
    \bra{x',y'}\ket{x,y}=\delta(x'-x)\delta(y'-y)  \qquad \text{and}\qquad
 \int_{0\le x\le  1-y \le 1}dxdy \ket{x,y}\bra{x,y}|=1.
\end{equation}

\begin{rem}
Since $z\ge 0$, it follows that
\[
x+y\le 1.
\]
Hence
\[
x\le 1-y.
\]
Because $x\ge 0$ and $y\ge 0$, we also have
\[
0\le x\le 1-y\le 1.
\]
In particular, each of the variables $x,y,z$ takes values in the unit interval and $(x,y,z)$ lies in the standard $2$–simplex.
\end{rem}

\subsection{The wavefunctions \label{sec:wave}}

In these simplex coordinates, the wavefunctions will be:
\begin{equation}
    \phi_{n,k}^{(a,b,c)}(x,y)= \langle x,y;a,b,c|n,k;a,b,c\rangle.
\end{equation}
As a consequence of the identity
\begin{equation}
    \langle x,y;a,b,c|\ X^{(1)}\ket{n,k;a,b,c}=\left(\langle x, y;a,b,c|X^{(1)}\right)\ket{n,k;a,b,c}=x\langle x,y;a,b,c\ket{n,k;a,b,c},
\end{equation}
and of a similar one with $X^{(1)}$ replaced by $X^{(3)}$ and the factor $x$ in the right-most expression by $(1-x-y)$, these wavefunctions satisfy the three-term relation given by
\begin{equation} \label{rec1}
    x \phi_{n,k}^{(a,b,c)}(x,y)=\alpha^1_{nk} \phi_{n+1,k}^{(a,b,c)}(x,y) + \alpha^2_{nk}  \phi_{n,k}^{(a,b,c)}(x,y) +\alpha^3_{nk} \phi_{n-1,k}^{(a,b,c)}(x,y),  
\end{equation}
where $\alpha^i_{nk},\ i=1,2,3$ are defined in Proposition \ref{prop:X1Action}, and satisfy also the nine-term relation  
\begin{align} \label{rec2}
    (1-x-y)\phi_{n,k}^{(a,b,c)}(x,y)&=\beta^1_{nk}\phi_{n+1,k-1}^{(a,b,c)}(x,y)+\beta^2_{nk} \phi_{n+1,k}^{(a,b,c)}(x,y)+\beta^3_{nk} \phi_{n+1,k+1}^{(a,b,c)}(x,y)\\
&\ \ +\beta^4_{nk}\phi_{n,k-1}^{(a,b,c)}(x,y) +  \beta^5_{nk} \phi_{n,k}^{(a,b,c)}(x,y)+\beta^6_{nk}\phi_{n,k+1}^{(a,b,c)}(x,y)\nonumber\\
&\ \ +\beta^7_{nk} \phi_{n-1,k-1}^{(a,b,c)}(x,y)+\beta^8_{nk} \phi_{n-1,k}^{(a,b,c)}(x,y)+\beta^9_{nk} \phi_{n-1,k+1}^{(a,b,c)}(x,y),\nonumber
\end{align}
with $\beta_{n,k}^i$ defined in Proposition \ref{prop:X3Action}.

These recurrence relations which are part of the representation theory of $\mathcal{J}_2$, are known \cite{crampe2025rank} to be verified by the two-variable Jacobi polynomials $J_{n.k}^{(a,b,c)}(x,y)$ on the triangle. Hence, since $J_{0.0}^{(a,b,c)}(x,y)=1$, 
\begin{equation} \label{wavefunctionxy}
    \phi_{n,k}^{(a,b,c)}(x,y)=\phi_{0,0}^{(a,b,c)}(x,y) J_{n.k}^{(a,b,c)}(x,y).
\end{equation}
In order to compute the ground state wavefunction $\phi_{0,0}^{(a,b,c)}(x,y)$, we shall use the explicit expression \eqref{posbas} of the position vectors $|x,y;a,b,c\rangle $ (noting that we will soon drop the parameters $a,b,c$ in their designation) as well as the vectors
\begin{equation}
    |n=0, k=0;a,b,c\rangle =\frac{1}{c_0}\sum_{m=0}^{\infty} q_m^{(\nu_{123})}(1) |\nu_{123}, \nu_{23};m\rangle.
\end{equation}
Recalling that $\langle \nu;j\ket{\nu;j'}=(2\nu)_jj!\delta_{j,j'}$ and the expression \eqref{eq:JointC123 and C23} for $   \ket{\nu_{123}+n,\nu_{23}+k;m}$,  we find
\begin{multline}
    \langle x,y\ket{0,0;a,b,c}=\frac{1}{c_0}\sum_{m=0}^\infty \sum_{j=0}^m\sum_{\ell=0}^{m-j}q_j^{(\nu_1)}(x)q_{m-j-\ell}^{(\nu_2)}(y)q_\ell^{(\nu_3)}(1-x-y)q_m^{(\nu_{123})}(1)\\
    \times\binom{m}{j}\binom{m-j}{m-j-\ell} j!(m-j-\ell)! \ell! (2\nu_1)_j(2\nu_2)_{m-j-\ell}(2\nu_3)_\ell.
\end{multline}
Simple combinatorics gives
\begin{equation}
    \binom{m}{j}\binom{m-j}{m-j-\ell} j!(m-j-\ell)! \ell!=m!
\end{equation}
and yields
\begin{multline}
    \langle x,y\ket{0,0;a,b,c}=\\
    \frac{1}{c_0}\sum_{m=0}^\infty m! \;q_m^{(\nu_{123})}(1)\sum_{j=0}^m\sum_{\ell=0}^{m-j}q_j^{(\nu_1)}(x)q_{m-j-\ell}^{(\nu_2)}(y)q_\ell^{(\nu_3)}(1-x-y)
    (2\nu_1)_j(2\nu_2)_{m-j-\ell}(2\nu_3)_\ell.
\end{multline}
Now insert the expressions for the functions $q_i^{(\nu)}$ that we reproduce below
\begin{equation}
    q_i^{(\nu)}(r)=\frac{(-1)^ i}{\Gamma(2\nu)(2\nu)_i} e^{-\frac{r}{2}}r^{\nu -\frac{1}{2}}L_i^{(2\nu -1)}(r).
\end{equation}
One gets
\begin{multline}
    \langle x,y\ket{0,0;a,b,c}=\\
    \frac{1}{c_0}\sum_{m=0}^\infty m! \;q_m^{(\nu_{123})}(1) \frac{e^{-\frac{1}{2}}x^{\nu_1-\frac{1}{2}}y^{\nu_2-\frac{1}{2}}(1-x-y)^{\nu_3-\frac{1}{2}}}{\sqrt{\Gamma(2\nu_1\Gamma(2\nu_2)\Gamma(2\nu_3)}}\sum_{j=0}^m\sum_{\ell=0}^{m-j}L_j^{(2\nu_1)}(x)L_{m-j-\ell}^{(2\nu_2)}(y)L_\ell^{(2\nu_3)}(1-x-y).
\end{multline}
Using twice the addition formula for the Laguerre polynomials
\begin{equation}
    \sum_{i=0}^nL_j^{(\alpha)}(x)L_{n-j}^{(\beta)}(y)=L_n^{(\alpha+\beta+1)}(x+y),
\end{equation}
one finds
\[
\sum_{j=0}^m \sum_{\ell=0}^{m-j}
L_j^{(2\nu_1-1)}(x)\,
L_{m-j-\ell}^{(2\nu_2-1)}(y)\,
L_\ell^{(2\nu_3-1)}(1-x-y) = L_m^{(2\nu_{123}-1)}(1).
\]
Putting this together, we have
\begin{equation} \label{xy00}
   \langle x,y|0,0;a,b,c\rangle=\frac{1}{c_0} \frac{e^{-\frac{1}{2}}x^{\nu_1-\frac{1}{2}}y^{\nu_2-\frac{1}{2}}(1-x-y)^{\nu_3-\frac{1}{2}}}{\sqrt{\Gamma(2\nu_1\Gamma(2\nu_2)\Gamma(2\nu_3)\Gamma(2\nu_{123})}}\sum_{m=0}^{\infty}\frac{m!}{(2\nu_{123})_m} \left(L_m^{(2\nu_{123}-1)}(1)\right)^2.
\end{equation}
Formally the sum in \eqref{xy00} is singular and per \eqref{resid} is equal to $e\Gamma(2\nu_{123})\delta(0)$. In the physical space spanned by the discrete vectors $\ket{n,k;a,b,c}$, we assume that $\delta(0)$ is regulated to be a constant $K$. Recalling from \eqref{cn} that 
\begin{equation}
    c_0=\left[\frac{\Gamma(2\nu_{123})}{\Gamma(2\nu_1\Gamma(2\nu_2)\Gamma(2\nu_3)}\right]^\frac{1}{2}, 
\end{equation}
we finally obtain
\begin{equation}
      \langle x,y|0,0;a,b,c\rangle=K x^{\frac{a}{2}}y^{\frac{b}{2}}(1-x-y)^{\frac{c}{2}}.
\end{equation}
There only remains to determine the constant $K$. This is done from the fact that we know on the one hand (see \eqref{scapr} and \eqref{normJac}) that the norm squared of $|0,0;a,b,c\rangle$  is given by
\begin{equation}
    \langle 0,0;a,b,c|0,0;a,b,c\rangle=N_0^{(a,b+c+1)}N_0^{(b,c)}=\frac{\Gamma(a+1)\Gamma (b+1) \Gamma (c+1)}{\Gamma(a+b+c+3)},
\end{equation}
and one the other hand, that 
\begin{align}
     &\int_{0\le x\le  1-y \le 1}dxdy \; \langle 0,0, a,b,c\ket{x,y}\bra{x,y}|0,0;a,b,c\rangle = \langle 0,0, a,b,c|0,0;a,b,c\rangle \nonumber\\
     &=|K|^2 \int_{0\le x\le  1-y \le 1}dxdy \; x^ ay^b(1-x-y)^c = |K|^2 \;\frac{\Gamma(a+1)\Gamma (b+1) \Gamma (c+1)}{\Gamma(a+b+c+3)},
\end{align}
 with the last equality known from the expression of the Dirichlet integral. This hence implies that $K=1$.
 It thus results that in the barycentric coordinates the wavefunctions are
 \begin{equation}\label{wavefuncbar}
      \phi_{n,k}^{(a,b,c)}(x,y)= \langle x,y;a,b,c|n,k;a,b,c\rangle = x^{\frac{a}{2}}y^{\frac{b}{2}}(1-x-y)^{\frac{c}{2}}J_{n.k}^{(a,b,c)}(x,y),
 \end{equation}
with $J_{n.k}^{(a,b,c)}(x,y)$ the two-variable Jacobi polynomials. 

We recall that these are polynomials of total degree $n$ in the variables $x$ and $y$, orthogonal on the simplex $0\le x \le 1-y \le1$, and are defined as follows \cite{Koornwinder}, \cite{dunkl2014orthogonal}, \cite{crampe2025rank}, \cite{crampe2026algebraic}:
\begin{align} 
 &J_{n,k}^{(a,b,c)}(x,y)=J_{n-k}^{(a,b+c+2k+1)}\bigl(x\bigr) \ (1-x)^k \ J_k^{(b,c)}\left(\frac{y}{1-x}\right), \quad \mbox{}\quad n \ge k \ge 0, \label{defJac2}
\end{align}
in terms of the univariate Jacobi polynomial $J_n^{(\alpha,\beta)}(u)$ on the interval $[0,1]$ given by
\begin{equation}
     J_n^{(\alpha,\beta)}(u) = \dfrac{(\alpha+1)_n}{n!}\ {}_2F_1 \argu{-n,n+\alpha+\beta+1}{\alpha+1}{x}. \label{defJac1}
\end{equation}
Their orthogonality property follows from in the present context from
\begin{align}
     &N_{n-k}^{(a,b+c+2k+1)}\,N_k^{(b,c)} \delta_{n'n}\delta_{k'k} = \langle n', k';a,b,c|n,k;a,b,c\rangle \nonumber\\
     =&\int_{0\le x\le  1-y \le 1}dxdy \; \langle n',k', a,b,c\ket{x,y}\bra{x,y}|n,k;a,b,c\rangle  \\
     =&\int_{0\le x\le  1-y \le 1}dxdy \; x^ ay^b(1-x-y)^c J_{n'.k'}^{(a,b,c)}(x,y)J_{n.k}^{(a,b,c)}(x,y).
\end{align}

\subsection{Spherical coordinates}

One might wish to return to the original physical/Cartesian coordinates $(x_1, x_2,x_3)$ in which the model was defined initially or, the spherical coordinates $(\vartheta, \varphi)$ given by
\begin{equation} \label{sph_coord}
    x_1=\sin\varphi \sin\vartheta, \quad x_2=\cos\varphi \sin\vartheta, \quad x_3=\cos\vartheta, \qquad 0 \le \vartheta \le \frac{\pi}{2}, \qquad 0 \le \varphi \le \frac{\pi}{2},
\end{equation}
where the range corresponds to the octant with  $x_1,x_2,x_3>0$. Recall that 
\begin{equation} \label{carbar}
    x=x_1^2, \qquad y=x_2^2,\qquad  z=1-x-y=x_3^2.
\end{equation}

The relation between the flat measure $dxdy$ on the triangle and the standard one, $d\mu_{S^2}=\sin\vartheta d\vartheta d\varphi$, on $S^2$ is readily found to be 
\begin{equation}
     dxdy=4\sqrt{xy(1-x-y)}\;d\mu _{S^2}.
\end{equation}
This implies that the orthogonal wave functions of $\cH$ for the usual measure on the sphere $\sin \theta~d\theta d\varphi$ are given by
\begin{equation}
    \Psi_{n,k}^{(a,b,c)}(\vartheta,\varphi)=x^{\frac{a}{2}+\frac{1}{4}}y^{\frac{b}{2}+\frac{1}{4}}(1-x-y)^{\frac{c}{2}+\frac{1}{4}}J_{n,k}^{(a,b,c)}(x,y),
\end{equation}
with $(x,y)=(x(\vartheta,\varphi),y(\vartheta,\varphi))$. 
This matches the formulas in \cite{Iliev17,Iliev18} where the generic superintegrable model on the two-sphere is discussed in the simplex coordinates. 

It is observed that the wavefunctions $\Psi^{(a,b,c)}_{n,k}(\vartheta,\varphi)$ labeled by the eigenvalues (involving $k$) of $L_1$, are not separated in the spherical coordinates. In this respect, it has been observed in \cite{crampe2026algebraic} that an alternative basis for the physical representation of $\mathcal{J}_2$ can be obtained by diagonalizing $L_3$. This leads to the functions
\begin{equation}
  ^{\pi}  \phi_{n,k}^{(a,b,c)}(x,y)
  =\left[x^ ay^b(1-x-y)^c\right]^{\frac{1}{2}} J_{n.k}^{(c,b,a)}(1-x-y,y),
\end{equation}
which is obtained from the wavefunction \eqref{wavefunctionxy} by performing the reflection $\pi$:
\begin{equation}
    x\leftrightarrow z=(1-x-y), \qquad a\leftrightarrow c.
\end{equation} It was also shown in \cite{crampe2026algebraic} that the overlaps between the functions $^{\pi}  \phi_{n,k}^{(a,b,c)}(x,y)$ and $\phi_{n,k}^{(a,b,c)}(x,y)$ are given in terms of Racah polynomials.
Now, it is readily found that
\begin{equation}
    1-x-y=\cos^2\vartheta,\qquad \frac{y}{x+y}=\cos^2\varphi.
\end{equation}
Hence, in the spherical coordinates \eqref{sph_coord}, the corresponding wavefunctions (with the inclusion of the Jacobian factor) will read:
\begin{align}
& ^{\pi}\Psi_{n,k}^{(a,b,c)}(\vartheta,\varphi)=\\
& (\sin\varphi \; \sin\vartheta )^{a+\frac{1}{2}} (\cos\varphi \; \sin\vartheta)^{b+\frac{1}{2}} (\cos \vartheta)^{c+\frac{1}{2}} \; J_{n-k}^{(c, a+b+2k+1)}(\cos^2\vartheta)\; \sin^{2k}\vartheta \; J_k^{(b,a)}(\cos^2\varphi),\nonumber
\end{align}
which are manifestly separated.
If we rewrite this wavefunction $^{\pi}\Psi_{n,k}^{(a,b,c)}(\vartheta,\varphi)$ in terms of the standard Jacobi polynomials $P_n^{(a, b)}(x)$ \cite{koekoek2010hypergeometric} which are related to $J_n^{(a,b)}(x)$ as follows
\begin{equation}
J_n^{(a,b)}(x) = P_n^{(a,b)}(1-2x),
\end{equation}
we find 
\begin{align} \label{wfPn}
& ^{\pi}\Psi_{n,k}^{(a,b,c)}(\vartheta,\varphi)=\\
& (\sin\varphi)^{a+\frac{1}{2}} (\cos\varphi)^{b+\frac{1}{2}} (\cos \vartheta)^{c+\frac{1}{2}} \; (\sin \vartheta) ^{a+b+2k+1} \; P_{n-k}^{(c, a+b+2k+1)}(-\cos2\vartheta)\;  P_k^{(b,a)}(-\cos2\varphi).\nonumber
\end{align}
Using the property \cite{functions1955bateman} according to which $P_n^{(a, b)}(-x)=(-1)^nP_n^{(b, a)}(x)$, we find that \eqref{wfPn} matches (up to an unimportant sign) with the expression for this wavefunction that is found in \cite{kalnins2007wilson} between eqs (8) and (9) if the (square root of the ) $\sin \vartheta$ of the measure $\mu_{S^2}$ is included in the function and if one makes the substitution $k \rightarrow n$ and $n \rightarrow k+n$.

Let us note that the wavefunctions $\Psi_{n,k}^{(a,b,c)}(\vartheta,\varphi)$ that we focused on in this paper, are separated in the \textit{different} spherical coordinates where $x_1$ and $x_3$ are exchanged (together with $a$ and $c$). This is in keeping with the fact \cite{miller2013classical} that superintegrable models with constants of motion that are quadratic in the momenta separate in more than one coordinate system and that in two dimensions these are in correspondance with the quadratic conserved quantity that is diagonalized in addition to the Hamiltonian.

\subsection{Differential properties of the two-variable Jacobi polynomials}

We have recalled in Section 2, the studies \cite{GVZ2014}, \cite{genest2014racah} that identify the generic superintegrable model on $S^2$ as the total Casimir operator subjected to a constraint, for the diagonal embedding of $\mathfrak{su}(1,1)$ realized through the sum of three singular oscillators. These studies allowed to identified the (differential) operators $\mathcal{L}, \mathcal{L}_1, \mathcal{L}_3, \mathcal{X}_1, \mathcal{X}_3$ that form the generators of the dynamical algebra. While the determination of the commutation relations of these operators is a priori straightforward given their explicit expressions, we have instead obtained those relations by using the $\mathfrak{su}(1,1)$ underpinning. This was then taken to model the abstract algebra $\mathcal{J}_2$ with generators correspondingly designated by straight letters and verifying the same relations. We then proceeded to construct the physical representation of $\mathcal{J}_2$ and to find that the two-variable Jacobi polynomials arise in the wavefunctions that are overlaps between two representation bases of $\mathcal{J}_2$.

We can now revert to the differential realization of the generators and exploit our representation theory results to recover in this superintegrable model framework the differential equations and recurrence relations obeyed by these bivariate polynomials. Precisely, we consider the expressions for the operators $\mathcal{L}, \mathcal{L}_1, \mathcal{L}_3, \mathcal{X}_1, \mathcal{X}_3$ as they arise from the model in the Cartesian coordinates $x_1, x_2, x_3$, we then pass to the coordinates $x, y$ given by \eqref{carbar} and perform the gauge transformation 
\begin{equation}
    W \rightarrow \widetilde{W}=  \left[\mathcal{G}^{(a,b,c)}(x,y)\right]^{-1}W\;\mathcal{G}^{(a,b,c)}(x,y) \qquad  \text{with} \qquad \mathcal{G}^{(a,b,c)}(x,y)= x^{\frac{a}{2}}y^{\frac{b}{2}}(1-x-y)^{\frac{c}{2}},
\end{equation}
for  $W=\mathcal{X}_1$, $\mathcal{X}_3$, $\mathcal{L}_1$, $\mathcal{L}_3$, $\mathcal{L}$ to find
\begin{align}
        &\widetilde{\mathcal{X}}_1=x,\\
    &\widetilde{\mathcal{X}}_3=1-x-y,\\
    &\widetilde{\mathcal{L}}_1=((b+1)(1-x)-(b+c+2)y)\partial_y + y(1-x-y) \partial_{yy},\label{L1}\\
    &\widetilde{\mathcal{L}}_3=((a+1)y-(b+1)x)(\partial_x-\partial_y) + xy \left(\partial_{xx}+\partial_{yy}-2\partial_{xy}\right)\label{L3t},\\
    &\widetilde{\mathcal{L}}= x(1-x)\partial_{xx}+y(1-y)\partial_{yy} \nonumber \\
    &\quad-2xy\partial_{xy}+(a+1-(a+b+c+3)x)\partial_{x}+(b+1-(a+b+c+3)y)\partial_{y}.
\end{align}
The properties of the two-variable Jacobi polynomials $J_{n,k}^{(a,b,c)}(x,y)$ defined in \eqref{defJac2} under the actions of these operators then follow from the representation that has been worked out. We have already given the two recurrence relations \eqref{rec1} and \eqref{rec2} that stem from the actions of  $\widetilde{X}_1$ and $\widetilde{X}_3$. In fact we have used them to identify the polynomials $J_{n,k}^{(a,b,c)}(x,y)$. From \eqref{actionLket} and \eqref{actionL1}, we have
\begin{align}
 &\widetilde{\mathcal{L}}\ [J_{n,k}^{(a,b,c)}(x,y)]=-n\,(n+a+b+c+2)J_{n,k}^{(a,b,c)}(x,y),\label{eqdiff1}\\
 &\widetilde{\mathcal{L}_1}\ [J_{n,k}^{(a,b,c)}(x,y)]=-k\,(k+b+c+1)J_{n,k}^{(a,b,c)}(x,y) \label{eqdiff2}.
\end{align}
Finally, Proposition \ref{prop:L3Action} implies that
\begin{equation} \label{eqdiff3}
    \widetilde{\mathcal{L}_3}\ [J_{n,k}^{(a,b,c)}(x,y)]=\gamma_{nk}^1 J_{n,k+1}^{(a,b,c)}(x,y)+\gamma_{nk}^2 J_{n,k}^{(a,b,c)}(x,y) +\gamma_{nk}^3 J_{n,k-1}^{(a,b,c)}(x,y),
\end{equation}
with the coefficients $\gamma_{nk}^i, i=1,2,3$ given by the equations \eqref{gamma1}, \eqref{gamma2}, \eqref{gamma3}. Given \eqref{eqdiff1}, \eqref{eqdiff2} and \eqref{eqdiff3} provide a realization of the centrally extended Racah algebra of rank one $\mathfrak{R}_1$ in the basis formed by the bivariate Jacobi polynomials.

\subsection{Additional remarks on the algebraic solution}

The algebraic solution of the generic superintegrable model on $S^2$ presented in this section proceeded by exploiting the representation of the dynamical algebra $\mathfrak{J}_2$. It relied on identifying the wavefunctions as the overlaps between the two representation bases $|n,k;a,b,c\rangle$ and $|x,y;a,b,c\rangle$ of $\mathfrak{J}_2$, and observing that the two recurrence relations of these overlaps entailed by the actions of the generators $X_1$ and $X_3$ in these bases coincide with those of the  bivariate Jacobi polynomials $J_{n,k}^{(a,b,c)}(x,y)$. This obviously assumes the a priori knowledge of this characterization of the polynomials $J_{n,k}^{(a,b,c)}(x,y)$ by their recurrence relations. We wish to add that there is a less learned approach to the algebraic solution of the generic superintegrable model on $S^2$ that is based more fundamentally on the representation theory of $\mathfrak{J}_2$. It calls upon the subalgebra structure of $\mathfrak{J}_2$ and involves as intermediary step the basic interpretation of the univariate Jacobi polynomials from the representation theory of the rank one Jacobi algebra $\mathfrak{J}_1$ \cite{granovskii1992mutual} \cite{genest2016tridiagonalization}. This method applies in the present context and is explained in detail in \cite{crampe2026algebraic}. In summary, it goes as follows.

A key element in this picture as mentioned, is the algebraic description of the univariate Jacobi polynomials $J_n^{(\alpha,\beta)}(u)$ as overlaps between the eigenbasis of the two (canonical) generators of the rank one Jacobi algebra $\mathcal{J}_1$. In a nutshell, with $n$ and $u$ denoting respectively the discrete and continuous eigenvalues of these operators and using the normalizations $\langle n',\alpha, \beta|n, \alpha, \beta\rangle = N_n^{(\alpha, \beta)} \delta_{n',n}$, and $\langle u'|u\rangle = \delta (u'-u)$, one has \cite{granovskii1992mutual, genest2016tridiagonalization,crampe2026algebraic}:
\begin{equation}
    \langle u, \alpha, \beta|n, \alpha, \beta\rangle = \left[ u^{\alpha} (1-u)^{\beta}\right]^{\frac{1}{2}} J_n^{(\alpha,\beta)}(u), 
\end{equation}
 with $\alpha$ and $\beta$ determined by the parameters of the algebra.

 It is standard to define representation bases elements as joint eigenvectors of maximal abelian subalgebras. Here, in addition to the two abelian algebras generated on the one hand by $\{X_1, X_3\}$ and, on the other hand by $\{L, L_1\}$ to which are respectively attached the eigenvectors $|x,y; a,b,c\rangle$ and $|n,k; a, b, c\rangle$, there is a third abelian subalgebra of relevance generated by $\{L_1, X_1\}$ to which are associated the eigenvectors $|x, k; b, c\rangle$ verifying
     \begin{equation} 
    X_1\;|x,k;b,c\rangle=x\;|x,k;b,c\rangle, \qquad L_1\;|x,k;b,c\rangle=-k(k+b+c)\;|x,k;b,c\rangle. \label{intbas}
\end{equation}
The situation, similar to the one encountered in the definition of the factorized Leonard pair \cite{FLP}, is depicted in Figure \ref{fig:1} where we have written the generators of these three abelian subalgebras over dots on a line segment. 
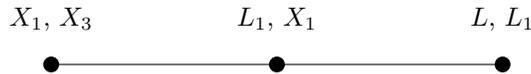
\begin{figure}[ht]
\centering
\begin{tikzpicture}[scale=1.5]
  \coordinate (A) at (0,0);
  \coordinate (B) at (2,0);
  \coordinate (C) at (4,0);

  \fill (A) circle (2pt);
  \fill (B) circle (2pt);
  \fill (C) circle (2pt);

  \draw (A) -- (B) -- (C);

  \node at ($(A)+(0,0.4)$) {\(X_1,\,X_3\)};
  \node at ($(B)+(0,0.4)$) {\(L_1,\,X_1\)};
  \node at ($(C)+(0,0.4)$) {\(L,\,L_1\)};
\end{tikzpicture}
\caption{This diagram illustrates how the two-variable Jacobi polynomials can be interpreted as convolutions of overlaps between representation bases of rank-one Jacobi algebra representations, connecting the \(\{\lvert x,y\rangle\}\) basis to the \(\{\lvert n,k;a,b,c\rangle\}\) basis via the \(\{\lvert x,k;b,c\rangle\}\) basis.}
\label{fig:1}
\end{figure}

Assuming that the vectors $|x, k;b,c\rangle$ are orthonormalized, we may use the resolution of the identity in terms of these basis vectors to write
\begin{equation}
    \langle x,y;a,b,c|n,k;a,b,c\rangle=\sum_{k'} \int _0^1 dx' \langle x, y;a,b,c\;|\;x',k';b,c\rangle\langle x',k';b,c\;|\;n,k; a, b, c\rangle. \label{conv1}
\end{equation}
Following the path in Figure \ref{fig:1} from left to right this amounts to expressing the wavefunction $\phi_{n,k}^{(a,b,c)}(x,y)$ as a convolution of the overlaps between the end bases and the middle one. Referring again to Figure \ref{fig:1} and the inner generators between pairs, it can be checked from the commutation relations given in Appendix \ref{sec:Rank2Jacobi} that the generators $X_3$ and $L_1$ form a Jacobi algebra of rank one where $X_1$ is central and where the continuous eigenvalue is $\frac{y}{1-x}$ and the parameters are $\alpha = b$, $\beta=c$. Similarly, it is also seen that the generators $X_1$ and $L$ form a rank one Jacobi algebra again where $L_1=-k(k+b+c+1)$ is central, the continuous eigenvalue is $x$ and the discrete one is $-(n-k)\left[(n-k)+a+b+c+2k+2\right]$ and the parameters are $\alpha=a, \;\beta=2k+b+c+1$. From these data it follows that 
\begin{equation}
       \langle x, y
       ; a, b. c\;|\;x',k'; b,c\rangle = \delta (x-x')\; \sqrt{\frac{1}{N_k^{(b,c)}}\frac{y^b(1-x-y)^c}{(1-x)^{b+c+1}}}\; J_{k'}^{(b,c)}\left(\frac{y}{1-x} \right), \label{overint1}
    \end{equation}
    and
    \begin{equation}
    \langle x',k'; b,c\;|\;n,k; a,b,c\rangle= \delta _{k,k'}\sqrt{{N_k^{(b,c)}x^a (1-x)^{(b+c+2k+1)}}}J_{n-k}^{(a, b+c+2k+1)}(x'). \label{overint2}
\end{equation}
Inserting these expressions in \eqref{conv1} yields the formula \eqref{wavefuncbar} for the wavefunctions.

This approach provides an algebraic solution of the superintegrable model on $S^2$ based on the dynamical algebra $\mathfrak{J}_2$ that is more refined than than the one given in Subsection \ref{sec:wave}. It uses the subalgebra structure of $\mathfrak{J}_2$ and takes the analysis to the level of the representation theory of $\mathfrak{J_1}$; it furthermore explains the intertwining of the univariate Jacobi polynomials in the expression of the two-variable ones. All the details are provided in \cite{crampe2026algebraic} where an even more elaborate pentagonal picture is presented.

\section{Conclusion}

Summarizing, this paper introduces, studies and uses the dynamical algebra of a maximally superintegrable model of great interest, namely the generic one on the two-sphere $S^2$. The symmetry algebra that accounts for the degenaracy of the energy levels of this system was known to be the Racah algebra of rank one $\mathfrak{R}_1$. Somewhat surprisingly, the dynamical algebra of this model had not been clearly identified so far. We showed that the spectrum generating algebra that incorporates all the state vectors of the system in a single irreducible module is the Jacobi algebra of rank two $\mathfrak{J}_2$ which includes as should be $\mathfrak{R}_1$ as subalgebra. This is a significant addition to the broad corpus of knowledge on superintegrable models.

Beyond identifying $\mathfrak{J}_2$ as the dynamical algebra of the generic superintegrable model on $S^2$, we further showed how an algebraic solution of the model is obtained from the construction of the physical representation of $\mathfrak{J}_2$. The steps of the analysis were the following:
\begin{enumerate}
    \item Start from the knowledge that $\mathfrak{su}(1,1)$ is the dynamical algebra of the one-dimensional singular oscillator. View the three-dimensional isotropic oscillator with three different singular terms as the diagonal embedding of $\mathfrak{su}(1,1)$ into its three-fold product. Recall from previous work \cite{genest2014racah}, \cite{genest2014generic} that the Hamiltonian $\mathcal{H}$ of the superintegrable model on $S^2$ corresponds to the total Casimir element, $\mathcal{C}^{(123)}$, in this differential realization of $\mathfrak{su}(1,1)$ in three variables $x_i,\ i=1,2,3$ supplemented with the constraint $x_1^2+x_2^2+x_3^2=1$. Take note as already observed, that the intermediate Casimir operators $\mathcal{C}^{(ij)}$ with $(ij) =(12), (23), (13)$ commute with the total Casimir operator and the constraint and, are hence conserved quantities. This implies that the differential operators $\mathcal{L}_1$ and $\mathcal{L}_3$, given as affine transformations of $\mathcal{C}^{(23)}$ and $\mathcal{C}^{(12)}$ respectively, realize the generators of the symmetry algebra of the model which is thus identified through this construct as the Racah algebra of rank one $\mathfrak{R}_1$ centrally extended by $\mathcal{L}$ affinely related to $\mathcal{C}^{(123)}$ (or equivalently $\mathcal{H}$). Observe then that this algebra $\mathfrak{R}_1$ can be extended to a dynamical algebra realized by adding to $\mathcal{L}$, $\mathcal{L}_1$, $\mathcal{L}_3$, the operators $\mathcal{X}_1=x_1^2$ and $\mathcal{X}_3 = x_3^2$.
    \item Use now the differential model as a guide to promote the operators $\mathcal{L}$, $\mathcal{L}_1$, $\mathcal{L}_3$,  $\mathcal{X}_1$, $\mathcal{X}_3$ to the abstract generators $L$, $L_1$, $L_3$, $X_1$, $X_s$ of the dynamical algebra by defining these by the same expressions in terms of the generators of $\mathfrak{su}(1,1)^{\otimes 3}$ as the one they have in the differential realization of this algebra. To be clear, with the $\mathfrak{su}(1,1)$ generators denoted by $\{J_0, J_{\pm}\}$ and represented by $\{\mathcal{J}_0, \mathcal{J}_{\pm}\}$ in the differential singular oscillator model in one variable $x$, we took for examble $X_1=\left[2J_0+J_{+}+J_{-}\right] \otimes 1\otimes 1$ since $\left[2\mathcal{J}_0+\mathcal{J}_{+}+\mathcal{J}_{-}\right] =x^2$. Given this embedding of the dynanical algebra into  $\mathfrak{su}(1,1)^{\otimes 3}$, the properties of this dynamical algebra are inferred from those of $\mathfrak{su}(1,1)$. The commutation relations are thus readily obtained from those of $\mathfrak{su}(1,1)$ and we identifiesdthe dynamical algebra with the rank two Jacobi algebra $\mathfrak{J}_2$ whose defining relations had been previously given \cite{crampe2025rank}.
    \item Our next task was to identify the relevant representation of $\mathfrak{J}_2$. A natural basis was first constructed as joint solutions on the $\mathfrak{su}(1,1)^{\otimes 3}$ module of the eigenvalue problems associated to $L, L_1$ and the constraint $X_1 + X_2 +X_3 \sim 1$.  With the total and intermediate Casimir elements being operators to be diagonalized, this involved using $\mathfrak{su}(1,1)$ recoupling coefficients. The diagonalization of one of the operators $X_1, X_2, X_3$, leads to generalized vectors given as sums over the standard discrete series basis elements of the corresponding factor with Laguerre polynomials as coefficients. This allowed to enforce the Dirac constraint. At this point we were equipped to proceed with the construction of the representation in that basis. The actions of the generators not involved in the eigenvalue problems were determined using their expressions in terms of the $\mathfrak{su}(1,1)^{\otimes 3}$ elements and $\mathfrak{su}(1,1)$ representation theory. This involved using the properties of the various polynomials that intervene (Hahn, Racah, Laguerre) and in particular contiguity relations for some of these functions that have been recently derived \cite{Contiguity25}.
    \item The last step consisted in showing how the representation of $\mathfrak{J}_2$ thus constructed provides an algebraic solution of the superintegrable model on $S^2$. To that end we introduced another basis given as the joint eigenfunctions of the ``position'' operators $X_1$ and $X_2$ (and automatically of $X_3$ given the constraint). Providing the wavefunctions then amounts to determining the overlaps between this position basis and the energy basis described in the preceding paragraph. An immediate outcome of the representation theory of $\mathfrak{J}_2$ is that the wavefunctions satisfy the two recurrence relations of the two-variable Jacobi polynomials. There then only remained to compute the ground state wavefunction which again could be done using results obtained in the construction of the representation of $\mathfrak{J}_2$

\end{enumerate}

Thinking about follow-ups to this study, it should be said that the generalization to superintegrable models on $S^n$ with $n \ge 3$ can in principle be achieved by extending the strategy presented here for $n=2$. The problem can be similarly set in the framework of $\mathfrak{su}(1,1)^{\otimes (n+1)}$ i.e. the $(n+1)-$ fold product of $\mathfrak{su}(1,1)$. The symmetry algebra is the Racah algebra of rank $n$, $\mathfrak{R}_n$, and the wavefunctions are expressed in terms of multivariate Jacobi polynomials. Some work on these models and the representations of the associated symmetry algebras has already been done in particular for $n=3$ \cite{kalnins2011two, genest2014generic,crampe2023representations, post2024racah} and for arbitrary $n$ as well \cite{ de2018higher, gaboriaud2018generalized, de2020racah}. It is expected that it should be possible to introduce a Jacobi algebra of rank $n+1$ that would have $\mathfrak{R}_n$ as a subalgebra and would be the dynamical algebra of the superintegrable model on $S^n$. Repeating the program carried out here for arbitrary $n$ would certainly be worthwhile in light of the results reported in \cite{Iliev18} providing the irreducible actions of the symmetry operators of these models in representation bases made out of Jacobi polynomials in $n$ variables.

Another project that we intend to pursue, is to determine and exploit the dynamical algebra of the superintegrable model on $S^2$ with reflections \cite{ genest2014bannai, de2015bannai, genest2015laplace}. In that case, the symmetry algebra is known to be the Bannai--Ito algebra \cite{tsujimoto2012dunkl}, which plays for the reflection model the role that the Racah algebra plays for the generic model. We plan on looking for the dynamical algebra of this model with reflections that would extend the Bannai--Ito algebra, in the same way that the Jacobi algebra of rank two extends the Racah algebra in the present work. We hope to report soon on these questions.

\appendix

\section{Presentation of the rank 2 Jacobi algebra}\label{sec:Rank2Jacobi}Here are the relations between the five generators $L, L_1, L_3, X_1, X_3$ that define the rank two Jacobi algebra $\mathcal{J}_2$.
We denote by $I$ the identity element.
 Recall that 
 \begin{equation}
     [L,L_1]=0, \; [L,L_3]=0, \; [L_1,X_1]=0, \; [L_3,X_3]=0, \; [X_1,X_3]=0. \label{zero}
 \end{equation}
 
\subsection{Commutators involving $[L,X_1]$}
 \begin{align}
     [[L,X_1]\;,L\;]&=2\{X_1,L\}-2L+2L_1-(a+b+c+1)((a+b+c+3)X_1 -(a+1)I),\label{LX1L}\\
     [[L,X_1],L_1]&=0 \quad \text{ by Jacobi identity}, \\
     [[L,X_1],L_3]&=\{X_1,L+L_3\}+\{X_3-I,L-L_1\}\nonumber\\
     & \quad -(a+b+c+1)\left((a+1)(X_1+X_3-I)+(b+1)X_1\right),\\ 
     [[L,X_1],X_1]&=-2X_1^2+2X_1,\label{LX1X1}\\
     [[L,X_1],X_3]&=-2X_1X_3.\label{LX1X3}
 \end{align}

\subsection{Commutators involving $[L,X_3]$}

\begin{align}
    [[L,X_3]\;,L\;]&=2\{ X_3,L\}-2L+2 L_3
    -(a+b+c+1)((a+b+c+3)X_3 - (c+1)I),\label{LX3L}\\
    [[L,X_3],L_1]&=\{X_1-I, L- L_3\}+\{X_3,L+L_1\}
\nonumber\\
& \qquad -(a+b+c+1)\left((c+1)(X_1+X_3-I)+(b+1)X_3\right),
\\
    [[L,X_3],L_3]&=0 \quad \text{by Jacobi identity},\\
    [[L,X_3],X_1]&=
     [[L,X_1],X_3] \quad \text{ by Jacobi identity}, \\
    [[L,X_3],X_3]&=-2X_3^2+2X_3.\label{LX3X3}
\end{align}

\subsection{Commutators involving $[L_1,L_3]$}
\begin{align}
[[L_1,L_3],\;L\;]&=0 \quad \text{by Jacobi identity},\\
[[L_1,L_3],L_1]&=2\{L_1,L_3\}+2L_1^2-2L_1L 
    +(b+c)(b+1)(L-L_1-L_3) \nonumber\\&
    - (b+c)(c+1)L_3-(b-c)(a+1)L_1, \label{L1L3L1}\\
[[L_1,L_3],L_3]&=-2\{L_1,L_3\}-2L_3^2+2L_3L  
    -(a+b)(b+1)(L-L_1-L_3) \nonumber\\&
    + (a+b)(a+1)L_1+(b-a)(c+1)L_3, \label{L1L3L3}\\
[[L_1,L_3],X_1]&=   -\{X_1-I,L-L_1-L_3\}-\{X_3,L-L_1\}
\nonumber\\
    &+(a+1)\left((b+1)X_3+(c+1)(X_1+X_3-I)\right), \\
[[L_1,L_3],X_3]&=\{X_1,L-L_3\} +\{X_3-I,L-L_1-L_3\}
\nonumber \\
    & 
-(c+1)\left((a+1)(X_1+X_3-I)+(b+1)X_1\right).
 \end{align}
 
\subsection{Commutators involving $[L_1,X_3]$}
\begin{align}
 [[L_1,X_3]\;,L\; ]&=
    [[L,X_3],L_1]\quad \text{ by Jacobi identity}, \\
 [[L_1,X_3],L_1]&=2\{X_3 ,L_1\}+\{X_1,L_1\}-2L_1  \nonumber \\
 &-(b+c) \left((b+c+2)X_3 -(c+1)(I-X_1) \right), \label{L1X3L1}\\
 [[L_1,X_3],L_3]&=
     [[L_1,L_3],X_3]\quad \text{ by Jacobi identity},\\
 [[L_1,X_3],X_1]&=0\quad \text{ by Jacobi identity},\\
 [[L_1,X_3],X_3]&=-2X_3^2+2(I-X_1)X_3.\label{L1X3X3}
\end{align}

\subsection{Commutators involving $[L_3,X_1]$}
\begin{align}
    [[L_3,X_1]\;,L\;]
     &=[[L,X_1],L_3]\quad \text{ by Jacobi identity},\\
    [[L_3,X_1],L_1]&=
   -[[L_1,L_3], X_1]\quad \text{ by Jacobi identity},\\
    [[L_3,X_1],L_3]&=\{2X_1+X_3-I,L_3\} \nonumber\\ 
 &-(a+b)\left((a+1)(X_1+X_3-I)+(b+1)X_1\right), \label{L3X1L3}\\
    [[L_3,X_1],X_1]&=-2X_1(X_1+X_3-I), \label{L3X1X1}\\
    [[L_3,X_1],X_3]&=0\quad \text{ by Jacobi identity}.
\end{align}
Note that the commutators down the list that can be obtained from the Jacobi identity $[[X,Y],Z]+[[Y,Z],X]+[[Z,X],Y] = 0$ are not repeated; rather, it is indicated to which expression previously given they are equal.

\subsection{Implied relations.}

Any commutation relations between two commutators can be computed from the previous relations. For example, let use the Jacobi identity and relations \eqref{LX1L}, \eqref{LX1X3} to transform the commutator $[[L,X_1],[L,X_3]]$:
\begin{align}
   &[[L,X_1],[L,X_3]]=[[[L,X_1],L],X_3]+[L,[[L,X_1],X_3]] \nonumber\\
   =&2[L_1,X_3]-2[L,X_3]+ 2[L,X_3] X_1-2[L,X_1]X_3.
\end{align}
Similar relations can be obtained for all the commutators between two commutators of the generators. These relations, together the defining ones, allows us to order the generators $L,L_1,L_3,X_1,X_3$ and the commutators $[L,X_1]$, $[L,X_3]$, $[L_1,L_3]$, $[L_1,X_3]$, $[L_3,X_1]$.

Let us mention that there exist additional relations between the ordered monomials. This feature has been already studied in \cite{crampe2021racah,post2024racah} for the higher rank Racah algebra. For example, the previous relation can be computed differently:
\begin{align}
   &[[L,X_1],[L,X_3]]=[[L,[L,X_3]],X_1]+[L,[X_1,[L,X_3]]] \nonumber\\
   =&2[L,X_1] -2[L_3,X_1]+2[L,X_1]X_3-2[L,X_3]X_1.
\end{align}
Comparison of the both results leads to 
\begin{align}
  [L_1,X_3]-[L,X_3]=[L,X_1] -[L_3,X_1]. \label{identity}
\end{align}

\section{Some properties of certain orthogonal polynomials families\label{app:A}}
In this appendix we collect necessary information about Laguerre, Racah and Hahn polynomials.

\subsection{Laguerre polynomials}
\label{sec:Laguerre-pols}
The Laguerre polynomials $L_n^{(\alpha)}(x)$ with parameter $\alpha \in \mathbb{C}$ are defined by
\begin{equation}\label{eq:LaguerreDefinition}
    L_n^{(\alpha)}(x)=\frac{(\alpha+1)_n}{n!}{}_1F_1\left({{-n}\atop
{\alpha+1}}\;\Bigg\vert \; x\right)\, .
\end{equation}
They satisfy the following three-term recurrence relation
\begin{equation}\label{eq:LaguerreRecurrence}
    x L_n^{(\alpha)}(x)=-(n+1) L_{n+1}^{(\alpha)}(x)+(2n+\alpha+1) L_n^{(\alpha)}(x)-(n+\alpha) L_{n-1}^{(\alpha)}(x).
\end{equation}
Their orthogonality relation takes the form
\begin{equation}
    \int_0^\infty e^{-x}x^{\alpha}L_m^{(\alpha)}(x)L_n^{(\alpha)}(x)dx=\Gamma(\alpha+1) \frac{(\alpha+1)_n}{n!} \delta_{mn},
\end{equation}
and the corresponding resolution of the identity is
\begin{equation}
\sum_{n=0}^{\infty}
\frac{n!}{(\alpha+1)_n}
L_n^{(\alpha)}(x)L_n^{(\alpha)}(y)
=
\frac{\Gamma(\alpha+1)}{e^{-x}x^{\alpha}}\,
\delta(x-y). \label{resid}
\end{equation}
They satisfy the following contiguity relations \cite{DLMF} 
\begin{align}
    L_n^{(\alpha)}(x)&=L_n^{(\alpha+1)}(x)-L_{n-1}^{(\alpha+1)}(x),\label{eq:contiguityLaguerre1}\\
    xL^{(\alpha)}_{n-1}(x)&=-nL_n^{(\alpha-1)}(x)+(n+\alpha-1)L_{n-1}^{(\alpha-1)}(x).\label{eq:contiguityLaguerre2}
\end{align}

\subsection{Racah polynomials}
\label{sec:Racah-pols}
The Racah polynomials $R_n(\lambda(x))=R_n(\lambda(x);\alpha,\beta,\gamma,\delta)$  are defined by
\begin{equation}
    R_n(\lambda(x))={}_4F_3 \argu{-n,n+\alpha+\beta+1,-x,x+\delta+\gamma+1}{\alpha+1,\beta+\delta+1,\gamma+1}{1},\ \ \ \ \ n=0,\cdots,N,
\end{equation}
with $\alpha+1=-N$, $\beta+\delta+1=-N$, or $\gamma+1=-N$ and 
\begin{equation}
    \lambda(x)=x(x+\gamma+\delta+1).
\end{equation}
They satisfy the following three term recurrence relation 
\begin{equation}\label{eq:Racah_Recurrence}
    \lambda(x)R_n(\lambda(x))=A_nR_{n+1}(\lambda(x)) -(A_n+C_n)R_n(\lambda(x))+C_nR_{n-1}(\lambda(x)),
\end{equation}
where 
\begin{align}
    A_n&=\frac{(n+\alpha+1)(n+\alpha+\beta+1)(n+\beta+\delta+1)(n+\gamma+1)}{(2n+\alpha+\beta+1)(2n+\alpha+\beta+2)} ,\\
    C_n&=\frac{n(n+\alpha+\beta-\gamma)(n+\alpha-\delta)(n+\beta)  }{(2n+\alpha+\beta+1)(2n+\alpha+\beta)}.
\end{align}

\subsection{Hahn polynomials}
\label{sec:Hahn-pols}
The Hahn polynomials $p_i(x;\rr)$ with $\rr=(\alpha,\beta,N)$ are defined by
\begin{align}
    p_i(x,\rr)={}_3F_2 \left({{-i,\;i+\alpha+\beta+1, \;-x}\atop
{\alpha+1,\;-N}}\;\Bigg\vert \; 1\right)\,,
\end{align}
where the parameters $\alpha$ and $\beta$ are real and $N$ a non-negative integer.
They satisfy the following orthogonality relation for $-1<\alpha,\beta$, or $\alpha,\beta<-N$
\begin{equation}\label{eq:orthogonality Hahn}
    \sum_{x=0}^N \binom{\alpha+x}{x}\binom{\beta+N-x}{N-x}p_i(x,\rr)p_j(x,\rr)=\frac{(-1)^i(i+\alpha+\beta+1)_{N+1}(\beta+1)_ii!}{(2i+\alpha+\beta+1)(\alpha+1)_i(-N)_iN!}\delta_{ij}.
\end{equation}
They also verify the following three-term recurrence relation
\begin{align} \label{eq:recuR}
& -x p_i(x; \rr)=A_{i,\rr}\; p_{i+1}(x; \rr)
-(A_{i,\rr}+C_{i,\rr})\;p_i(x; \rr)+C_{i,\rr}\; p_{i-1}(x; \rr)\,,
\end{align}
with $p_0(x;\rr)=p_i(0;\rr)=1$. 
The coefficients of the recurrence relation read
\begin{align}
    A_{i,\rr}=\frac{(i+\alpha+\beta+1)(i+\alpha+1)(N-i)}{(2i+\alpha+\beta+1)(2i+\alpha+\beta+2)}\,,\quad
    C_{i,\rr}=\frac{i(i+\alpha+\beta+N+1)(i+\beta)}{(2i+\alpha+\beta)(2i+\alpha+\beta+1)}\,.
\end{align}
The Hahn polynomials satisfy contiguity relations. The $B_2$-contiguity relations take the following form
\begin{equation}\label{eq:lp}
\lambda^{+}_{x; \rr}\; p_i(x; \rr)= \Phi^{+1,+}_{i}\ p_{i+1}(\ox; \orr)+\Phi^{0,+}_{i}\ p_{i}(\ox; \orr)+\Phi^{-1,+}_{i}\ p_{i-1}(\ox; \orr)\,,\\
\end{equation} 
where $\ox$ is a linear transformation of $x$ and $\orr$ is the list of modified parameters in $\rr$.
Similarly the $B_2'$-contiguity relations relations take the following form
\begin{subequations}
\begin{align}
\lambda^+_{x,\rr} p_{i}(x;\rr)&=\Phi_{i}^{0,+}p_{i}(\ox;\orr) +\Phi_{i}^{-1,+}p_{i-1}(\ox;\orr)+\Phi_{i}^{-2,+}p_{i-2}(\ox;\orr)\,, \qquad i\geq 0\,, \\
\lambda^-_{x,\rr} p_{i}(\ox;\orr)&=\Phi_{i}^{0,-}p_{i}(x;\rr) +\Phi_{i}^{1,-}p_{i+1}(x;\rr)+\Phi_{i}^{2,-}p_{i+2}(x;\rr)\,, \qquad i\geq 0\,. 
\end{align}
\end{subequations} 
In \cite{Contiguity25}, the $B_2$ and $B_2'$-contiguity relations
for the finite families of orthogonal polynomials in the Askey scheme have been classified. Next we recall the ones used in this paper, following the same labeling of the solutions as in \cite{Contiguity25}.\\

\noindent
\textit{$B_2$-contiguity relations}
\begin{enumerate}

    \item[(HI/II)]$\ox=x+1, \quad \oalpha=\alpha,\quad \obeta=\beta\,,\quad \oN=N+1$
    \begin{align*}
 &  \lambda_{x,\rr}^+=x+1\,,
 \\
  &  \Phi_i^{+1,+}=-\frac{(N+1)(i+\alpha+1)(i+\beta+\alpha+1)}{(2i+\alpha+\beta+1)(2i+\alpha+\beta+2)}\,,
  \\
    &  \Phi_i^{-1,+}=-\frac{i(N+1)(i+\beta)}{(2i+\alpha+\beta+1)(2i+\alpha+\beta)}\,,
    \\
     &  \Phi_i^{0,+}=\frac{N+1 }{2i+\alpha+\beta+1}\left(\frac{(i+\alpha)(i+\alpha+\beta)}{2i+\alpha+\beta}+\frac{(i+1)(i+\beta+1)}{2i+\alpha+\beta+2} \right).
     \end{align*}
    \item[(HII/I)]$\ox=x-1, \quad \oalpha=\alpha,\quad \obeta=\beta\,,\quad \oN=N-1$
     \begin{align*}
     &  \lambda_{x,\rr}^+=x+\alpha\,,
    \\
     &  \Phi_i^{+1,+}=-\frac{(N-i)(N-1-i)(i+\alpha+1)(i+\beta+\alpha+1)}{(2i+\alpha+  \beta+1)(2i+\alpha+\beta+2)N}\,,
     \\
      &  \Phi_i^{-1,+}=-\frac{i(i+\beta)(i+\alpha+\beta+N)(i+\alpha+ \beta+N+1)}{(2i+\alpha+  \beta+1)(2i+\alpha+\beta)N}\,,
       \\
     &  \Phi_i^{0,+}=\frac{(N-i)(i+\alpha+\beta+N+1)}{(2i+\alpha+\beta+1)N}\left(\frac{(i+\alpha)(i+\alpha+\beta)}{2i+\alpha+\beta}+\frac{(i+1)(i+\beta+1)}{2i+\alpha+\beta+2} \right).
     \end{align*}
     \item[(HIII/IV)] $\ox=x, \quad \oalpha=\alpha,\quad \obeta=\beta\,,\quad \oN=N+1$
\begin{align*}
 &  \lambda_{x,\rr}^+=N+1-x\,,
 \\
  &  \Phi_i^{+1,+}=\frac{\left(i +1+\alpha +\beta \right) \left(N +1\right) \left(i +1+\alpha \right)}{\left(2 i +1+\alpha +\beta \right) \left(2 i +2+\alpha +\beta \right)}\,,
  \\
    &  \Phi_i^{-1,+}=\frac{i \left(N +1\right) \left(i +\beta \right)}{\left(2 i +1+\alpha +\beta \right) \left(2 i +\alpha +\beta \right)}\,,
    \\
     &  \Phi_i^{0,+}=\frac{\left(\alpha  \beta +2 \alpha  i +\beta^{2}+2 i \beta +2 i^{2}+\alpha +\beta +2 i \right) \left(N +1\right)}{\left(2 i +\alpha +\beta \right) \left(2 i +2+\alpha +\beta \right)}
\,.
\end{align*}
    \item[(HIV/III)] $\ox=x, \quad \oalpha=\alpha,\quad \obeta=\beta\,,\quad \oN=N-1$
\begin{align*}
 &  \lambda_{x,\rr}^+=N+\beta-x\,,
 \\
  &  \Phi_i^{+1,+}=\frac{\left(i +\alpha +\beta +1\right) \left(N -i \right) \left(N -1-i \right) \left(i +\alpha +1\right)}{\left(2 i +1+\alpha +\beta \right) N \left(2 i +2+\alpha +\beta \right)}\,,
  \\
    &  \Phi_i^{-1,+}=\frac{i \left(i +1+\alpha +\beta +N \right) \left(i +\alpha +\beta +N \right) \left(i +\beta \right)}{\left(2 i +1+\alpha +\beta \right) N \left(2 i +\alpha +\beta \right)}\,,
    \\
     &  \Phi_i^{0,+}=\frac{\left(\beta  \alpha +2 i \alpha +\beta^{2}+2 i \beta +2 i^{2}+\alpha +\beta +2 i \right) \left(N -i \right) \left(i +1+\alpha +\beta +N \right)}{\left(2 i +\alpha +\beta \right) N \left(2 i +2+\alpha +\beta \right)}
\,.
\end{align*}
\end{enumerate}

\noindent
\textit{$B_2'$-contiguity relations}
\begin{enumerate}
  \item[(HIII/III)] $\ox=x, \quad \oalpha=\alpha,\quad \obeta=\beta+2\,,\quad \oN=N$
\begin{align*} &  \lambda_{x,\rr}^+=1\,,
 \\
  &  \Phi_i^{0,+}=\frac{(i+\alpha+\beta+1)(i+\alpha+\beta+2)}{\left(2 i +1+\alpha +\beta \right)  \left(2 i +2+\alpha +\beta \right)}\,,
  \\
    &  \Phi_i^{-1,+}=\frac{2i(i+\alpha+\beta+1)}{(2i+\alpha+\beta)(2i+\alpha+\beta+2)}\,,
    \\
     &  \Phi_i^{-2,+}=\frac{i(i-1)}{(2i+\alpha+\beta)(2i+\alpha+\beta+1)}
\,.
  \\
  \\
 &  \lambda_{x,\rr}^-=(N+\beta-x+1)(N+\beta-x+2)\,,
 \\
  &  \Phi_i^{0,-}=\frac{\left(i +\beta +1\right) (i+\beta+2)(i+\alpha+\beta+N+3)(i+\alpha+\beta+N+2)}{\left(2 i +2+\alpha +\beta \right)  \left(2 i +3+\alpha +\beta \right)}\,,
  \\
    &  \Phi_i^{1,-}=\frac{2(i+\beta+2)(i+\alpha+1)(N-i)(i+\alpha+\beta+N+3)}{(2i+\alpha+\beta+2)(2i+\alpha+\beta+4)}\,,
    \\
     &  \Phi_i^{2,-}=\frac{(N-i)(N-i-1)(i+\alpha+1)(i+\alpha+2)}{(2i+\alpha+\beta+3)(2i+\alpha+\beta+4)}
\,.
  \\
\end{align*}
    \item[(HIV/III)] $\ox=x, \quad \oalpha=\alpha,\quad \obeta=\beta+2\,,\quad \oN=N-1$
\begin{align*}
 &  \lambda_{x,\rr}^+=1\,,
 \\
  &  \Phi_i^{0,+}=\frac{\left(i +1+\alpha +\beta \right) \left(N -i \right) \left(i +2+\alpha +\beta \right)}{\left(2 i +1+\alpha +\beta \right) N \left(2 i +2+\alpha +\beta \right)}\,,
  \\
    &  \Phi_i^{-1,+}=\frac{\left(\alpha +2 N +\beta +2\right) \left(i +1+\alpha +\beta \right) i}{\left(2 i +\alpha +\beta \right) \left(2 i +2+\alpha +\beta \right) N}
\,,
    \\
     &  \Phi_i^{-2,+}=\frac{i \left(i +1+\alpha +\beta +N \right) \left(i -1\right)}{\left(2 i +1+\alpha +\beta \right) N \left(2 i +\alpha +\beta \right)}
\,.
\end{align*}
\item[(HIII/IV)] $\ox=x, \quad \oalpha=\alpha,\quad \obeta=\beta+2\,,\quad \oN=N-1$
\begin{align*}
 &  \lambda_{x,\rr}^-=(N-x)(N+\beta-x+1)\,,
 \\
  &  \Phi_i^{2,+}=\frac{N \left(i +1+\alpha \right) \left(N -1-i \right) \left(i +2+\alpha \right)}{\left(2 i +3+\alpha +\beta \right) \left(2 i +4+\alpha +\beta \right)}\,,
  \\
    &  \Phi_i^{1,+}=
\frac{\left(\alpha +2 N +\beta +2\right) \left(i +1+\alpha \right) N \left(i +2+\beta \right)}{\left(2 i +2+\alpha +\beta \right) \left(2 i +4+\alpha +\beta \right)}
\,,
    \\
     &  \Phi_i^{0,+}=\frac{N \left(i +1+\beta \right) \left(i +2+\alpha +\beta +N \right) \left(i +2+\beta \right)}{\left(2 i +2+\alpha +\beta \right) \left(2 i +3+\alpha +\beta \right)}
\,.
\end{align*}
\item[(HIV/IV)] $\ox=x, \quad \oalpha=\alpha,\quad \obeta=\beta+2\,,\quad \oN=N-2$
\begin{align*}
 &  \lambda_{x,\rr}^+=1\,,
 \\
  &  \Phi_i^{0,+}=\frac{(N-i)(N-i-1)(i+\alpha+\beta+1)(i+\alpha+\beta+2)}{N(N-1)(2i+\alpha+\beta+1)(2i+\alpha+\beta+2)}\,,
  \\
    &  \Phi_i^{-1,+}=\frac{2i(N-i)(i+\alpha+\beta+1)(i+\alpha+\beta+N+1)}{N(N-1)(2i+\alpha+\beta+2)(2i+\alpha+\beta)}\,,
    \\
     &  \Phi_i^{-2,+}=\frac{i(i-1)(i+\alpha+\beta+N)(i+\alpha+\beta+N+1)}{N(N-1)(2i+\alpha+\beta+1)(2i+\alpha+\beta)}
\,.\\
\\
 &  \lambda_{x,\rr}^-=(N-x)(N-x-1)\,,
 \\
  &  \Phi_i^{0,-}=\frac{N(N-1)(i+\beta+1)(i+\beta+2)}{(2i+\alpha+\beta+2)(2i+\alpha+\beta+3)}\,,
  \\
    &  \Phi_i^{1,-}=\frac{2N(N-1)(i+\beta+2)(i+\alpha+1)}{(2i+\alpha+\beta+2)(2i+\alpha+\beta+4)}\,,
    \\
     &  \Phi_i^{2,-}=\frac{N(N-1)(i+\alpha+1)(i+\alpha+2)}{(2i+\alpha+\beta+3)(2i+\alpha+\beta+4)}
\,.
\end{align*}

\end{enumerate}


\bigskip

\section*{Acknowledgements}N.~Cramp\'e is partially supported by the international research project AAPT of the CNRS. L.~Vinet is funded in part by a Discovery Grant from the Natural Sciences and Engineering Research Council (NSERC) of Canada. Q. Labriet and L. Morey enjoy postdoctoral fellowships provided by this grant. AZ is supported by the Ministry of Science and Higher Education of the Russian Federation
(agreement no. 075–15–2025–343)

\section*{Conflict of interest}
On behalf of all authors, the corresponding author states that there is no conflict of interest.

\section*{Data availability}
This manuscript has no associated data.
\bibliographystyle{unsrt}
\bibliography{biblio}
\end{document}